\let\TeXyear\year
\documentclass[journal]{IEEEtran}

\let\year\TeXyear

\usepackage{cite}
\usepackage{amsmath,amssymb,amsfonts,amsthm}
\usepackage{bm}
\usepackage{algorithmic}
\usepackage{algorithm}
\usepackage{array}
\usepackage{textcomp}
\usepackage{stfloats}
\usepackage{float}
\usepackage[unicode]{hyperref}
\usepackage{xurl}
\usepackage{booktabs}
\usepackage{multirow}
\usepackage{colortbl}
\usepackage{tabularx}
\usepackage{makecell}
\usepackage{standalone}
\usepackage{microtype}
\usepackage{multibib}
\usepackage[T1]{fontenc} 
\usepackage{xspace}
\usepackage{listings}
\usepackage[autostyle=true]{csquotes}
\usepackage{etoolbox}
\usepackage{set}
\usepackage{maps}

\usepackage{todonotes}
\usepackage{spacedno}

\usepackage{tilde}

\definecolor{accessblue}{cmyk}{1 0.3 0 0.2}
\definecolor{greycolor}{cmyk}{0,0,0,.8}
\definecolor{bluelight}{cmyk}{1,0.3,0,0.2}

\newcites{r}{Reviewed Literature}

\newcommand{\bibdir}{./}
\graphicspath{{pictures/}{pictures/passport_photos/}}

\newcommand{\gitlab}{\texttt{GitLab}\xspace}

\usepackage{booktabs}
\usepackage{threeparttable,tabularx}
\usepackage[inline]{enumitem}

\usepackage{tikz}
\usetikzlibrary{arrows.meta, automata, positioning, bending, quotes, fit, shapes, calc}

\newcommand{\N}{\mathbb{N}}
\renewcommand{\theta}{\vartheta}

\usepackage{subcaption}
\newif\ifanonymous
\anonymousfalse

\newif\iffinal
\finalfalse

\iffinal
\hypersetup{hidelinks}
\newcommand{\TODO}[1]{}
\newcommand{\short}[1]{}
\newcommand{\FIXME}[1]{}
\newcommand{\stefan}[1]{}
\newcommand{\daniel}[1]{}
\newcommand{\sophia}[1]{}
\newcommand{\tobi}[1]{}

\newcommand{\mycomment}[3][blue]{}
\else
\newcommand{\TODO}[1]{\vspace{0.5em}\todo[inline, color=orange!30]{\textbf{TODO:} #1}}
\newcommand{\FIXME}[1]{\todo[size=\small, color=red!30]{\textbf{FIXME:} #1}}
\newcommand{\short}[1]{\todo[inline,color=gray!30,size=\small]{Content:#1}}
\newcommand{\stefan}[1]{{\color{blue!80}{stefan: #1}}}
\newcommand{\daniel}[1]{{\color{purple}{daniel: #1}}}
\definecolor{dartmouthgreen}{rgb}{0.05, 0.5, 0.06}
\newcommand{\sophia}[1]{{\color{dartmouthgreen}{sophia: #1}}}
\newcommand{\tobi}[1]{{\color{orange}{tobi: #1}}}

\newcommand{\mycomment}[3][blue]{\marginpar{\small\textcolor{#1}  {#2: #3}}}
\fi

\makeatletter
\newcommand*\rel@kern[1]{\kern#1\dimexpr\macc@kerna}
\newcommand*\widebar[1]{%
  \begingroup
  \def\mathaccent##1##2{%
    \rel@kern{0.8}%
    \overline{\rel@kern{-0.8}\macc@nucleus\rel@kern{0.2}}%
    \rel@kern{-0.2}%
  }%
  \macc@depth\@ne
  \let\math@bgroup\@empty \let\math@egroup\macc@set@skewchar
  \mathsurround\z@ \frozen@everymath{\mathgroup\macc@group\relax}%
  \macc@set@skewchar\relax
  \let\mathaccentV\macc@nested@a
  \macc@nested@a\relax111{#1}%
  \endgroup
}
\makeatother

\theoremstyle{plain}
\newtheorem*{fact}{Fact}

\newtheorem{theorem}{Theorem}
\newtheorem{lemma}{Lemma}
\newtheorem{corollary}{Corollary}

\theoremstyle{definition}
\newtheorem{definition}{Definition}

\usepackage{academicons}
\definecolor{orcidlogocol}{HTML}{A6CE39}
\newcommand{\orcid}[1]{\href{https://orcid.org/#1}{\includegraphics{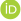}}}

\renewcommand{\orcid}[1]{}

\begin{document}


\title{On the Formalization of Cryptographic Migration}

\ifanonymous
\IEEEauthorblockN{Anonymized.}
\else
\author{
\IEEEauthorblockN{
Daniel Loebenberger\IEEEauthorrefmark{1}\IEEEauthorrefmark{3}\orcid{0000-0002-7969-6260},
Eduard Hirsch\IEEEauthorrefmark{3}\orcid{0000-0001-9593-2342},
Stefan-Lukas Gazdag\IEEEauthorrefmark{2}\orcid{0009-0002-1338-9484},
Daniel Herzinger\IEEEauthorrefmark{2}\orcid{0009-0009-3775-9567},
Christian Näther\IEEEauthorrefmark{4}\orcid{0009-0000-0634-817X},
Jan-Philipp Steghöfer\IEEEauthorrefmark{4}\orcid{0000-0003-1694-0972}~\IEEEmembership{Member, IEEE}
}

\IEEEauthorblockA{\IEEEauthorrefmark{1}Fraunhofer Institute for Applied and Integrated Security, Munich, Germany}
\IEEEauthorblockA{\IEEEauthorrefmark{2}genua GmbH, Kirchheim near Munich, Germany}
\IEEEauthorblockA{\IEEEauthorrefmark{3}OTH Amberg-Weiden, Amberg, Germany}
\IEEEauthorblockA{\IEEEauthorrefmark{4}XITASO GmbH, Augsburg, Germany}

\thanks{Manuscript created 31 July, 2025; This work was funded by the German Federal Ministry of Education and Research.}
}

\markboth{Loebenberger et al. ``On the Formalization of Cryptographic Migration''}{Loebenberger et al. ``On the Formalization of Cryptographic Migration''}

\fi

\maketitle

\begin{abstract}
We present a novel approach to gaining insight into the structure of cryptographic migration problems which are classic problems in applied cryptography. We use a formal model to capture the inherent dependencies and complexities of such transitions. Using classical mathematical results from combinatorics, probability theory, and combinatorial analysis, we evaluate the challenges of migrating large cryptographic IT infrastructures and prove that\,--\,in a suitable sense\,--\,cryptographic migration exhibits a certain expected complexity. We also provide numerical data for selected parameter sets. Furthermore, we analyze the proposed model in terms of real-world patterns and its practical applicability. Additionally, we discuss the challenges of modeling real-world migration projects. As concrete examples we examine the transition to post-quantum cryptography of the CI/CD system \gitlab and the multi-level technological transition of distribution power grids. This work paves the way for future advancements in both the theoretical understanding and practical implementation of cryptographic migration strategies.
\end{abstract}







\section{Introduction}
\label{sec:introduction}

\IEEEPARstart{C}{ryptography} is fundamental to ensuring information security in today's digital world and is thus an essential cornerstone of many modern dependable systems. Despite advances in cryptographic technologies, the process of migrating existing systems and IT infrastructures to new standards presents significant challenges~\cite{barpol2021nist}. One example is the transition from IPv4 to IPv6, which took more than 20 years and promoted technologies such as gateways and network address translation (NAT) to circumvent the actual migration. In the cryptographic realm, notable examples include updating key sizes (where RSA key sizes increased, as did integration problems) and transitioning between different hash functions (from SHA-1 to the SHA-2 family, and eventually to the SHA-3 family), with the SHA-3 family being adopted slowly as SHA-2 is still considered sufficient for most tasks. Rapid transition appears to require a strong incentive to adopt new technology. A notable exception is the transition from the symmetric cipher DES to its replacements AES, driven by the need for stronger security measures due to DES's vulnerability to brute-force attacks~\cite{Patil16}.
The issue of migration is particularly prominent in the current transition towards post-quantum cryptography (PQC), which aims to protect data against future quantum computing threats. The German Federal Office for Information Security (BSI) uses the working hypothesis for the high-security sector that there likely exists a cryptographically relevant quantum computer in the early 2030s~\cite{bsi24statusQC}. Mosca's famous inequality~\cite{mosca2018-will-we-be-ready} states that there is a 1 in 2 chance that such a computer will exist in 2031. Thus, there are numerous practical post-quantum migrations going on world-wide. The reason is that we are already running-out of time: Sensitive (e.\,g., classified) data often has long confidentiality requirements (e.\,g., 25 years) and it might be possible to decrypt cryptographically protected data which is recorded today in the future. In addition, migrating the IT-infrastructure correspondingly takes a considerable amount of time~\cite{pqc_adoption}. 
%

Our approach to support PQC migration is to model the migration problem suitably and identify this way optimal migration strategies in general. We propose a systematic analysis to better understand the structural problems and explore which potential models are necessary for building tools to assist actual PQC migrations. Although algorithmic aspects and the details of the concrete tooling are beyond the scope of this work, we believe that a structured assessment of the migration problem provides valuable insights into the structural complexties involved.

Our findings show that the dependencies of a system to be migrated usually necessitate a compatible migration order to avoid losing functionality. Within this order, there are always sets of systems which need to be migrated together\,--\,so called \emph{migration clusters} (cf.~Definition~\ref{def:migration-cluster})\,--\,and their successive migration defines the corresponding migration strategy. As soon as the systems under evaluation exhibit non-trivial complexity, we have many or large clusters, and these clusters are expected to be interdependent, making parallelization of the different steps impossible and migration thus difficult. To facilitate the process while still allowing for rapid migration, we need tools that can identify and visualize optimal migration strategies. In order to cope with the complexity of developing such tools, we believe that a proper understanding of the underlying problem is necessary. Our work presented here provides the formal basis for this.

\subsection*{Contributions}

Our contributions are as follows:
\begin{enumerate}
  \item We introduce a novel graph-based formalization of cryptographic migration, along with a discussion of the structure of the resulting model (Section~\ref{sec:formal_model}).
  \item We then prove that any (reasonably large) migration project has an expected complex dependency structure and thus show that this structure prohibits parallelizing the migration (Section~\ref{sec:migration_hard}).
  \item We explore typical graphs and patterns that occur in practice and how these patterns can support system administrators in planning a migration project (Section~\ref{sec:migration-modelling}).
  \item We demonstrate how our model can be applied to an actual migration problem as an exemplar to showcase its ability to draw conclusions about real-world migrations. We also discuss how multiple graphs on different levels of abstraction can be combined when planning the migration of a complex critical infrastructure system (Section~\ref{sec:real-world-examples}).
\end{enumerate}
 
Overall, the formalization proposed in this paper is an important step toward a better, more systematic understanding of migration problems. It also lays the groundwork for developing tools that will support engineers in the complex cryptographic migrations that lie ahead in the coming years.

\section{Background and Related Work}
\label{sec:related_work}

\subsection{Background}


After the announcements of the first post-quantum public key algorithms for consideration as an international standard by the National Institute of Standards and Technology (NIST)~\cite{fips203, fips204, fips205}, the transition to PQC has started in practice: in the US, there is a project run by NIST with a large consortium of industry-leading companies that explores the transition to PQC~\cite{barpol2021nist}. In Europe, the German Federal Office for Information Security (BSI) published a guideline concerning post-quantum migration~\cite{bsi_migration} and also included certain post-quantum algorithms in their recommended suite of cryptographic algorithms~\cite{bsi-tr-02102-1}. The Dutch government also published a ``PQC Migration Handbook'' containing guidelines for migrating to PQC~\cite{PQC-migration-handbook-nl2023}.
%
%
%

%
%
%
%
%
%
Despite these world-wide activities, we found little research on the general structure of migration projects and how to suitably model cryptographic migration. In fact, many people consider migrating soft- and hardware components as a best-practice task. So, people legitimately ask ``Where Is the Research on Cryptographic Transition and Agility?'' as David Ott did at RWC2022 (see also~\cite{ott2023research}). We show in Section~\ref{sec:migration_hard} that there are structural reasons why cryptographic migration is so hard to achieve in practice.

Ott et al.~define the term \emph{cryptographic agility} as ``the ability to make a transition from one cryptographic algorithm or implementation to
another''~\cite{ott2023research}. Our own, canonical definition~\cite{naether24sok} states that ``Cryptographic Agility is a theoretical or practical approach, objective, or property which provides capabilities for setting up, identifying, and modifying encryption methods and keying material in a flexible and efficient way while preserving business continuity.'' Adopting cryptographic agility in this sense ensures the flexibility to transition between algorithms or implementations in response to evolving threats, see Section~\ref{sec:real-world-examples:crypto-agility} for more details. Thus, in order to implement cryptographic agility we need to prepare cryptographic migration. In turn, for preparing migration, we need a step-by-step migration plan which must allow careful re-evaluation of dependencies and system integrity at each stage. This kind of approach is in line with making security solutions sustainable by allowing to adapt to changing threats and technologies.

\subsection{Related Work}

\ifanonymous Recent \else Our recent \fi efforts have focused on consolidating the state-of-the-art in cryptographic
migration~\cite{naether2024migrating} and systematically organizing knowledge on
cryptographic agi\-li\-ty~\cite{naether24sok}. Furthermore, frameworks arose that aim to
support the migration process as in~\cite{wiesmaier2021pqc}. Also, there are few
informal models of the migration problem (e.\,g.,~\cite{ott2019identifying, netwie24, schhen24, etsi24cyberqsc}).
All these models have in common that they address operational aspects of cryptographic
migration in various institutional contexts. This can be aligned with the above observation
of structuring best-practice tasks. While these kind of analyses are clearly important and
can serve as a basis for practical migration, they do not attempt to formally assess the
underlying structural problem. Instead they give guidelines on how to assess the needed transition efforts,
explain how to perform a suitable risk-analysis, and describe process landscapes to perform cryptographic transition.

One way of collecting the information that is necessary to build a migration model is to construct a software bill of materials (SBOM)~\cite{odonoghue2025softwarematerialssoftwaresupply} or a cryptographic bill of materials (CBOM)~\cite{leirimaa2024supporting,hess24cboms}. While the SBOM contains information about which software in which version is used and can also include information about dependencies that are relevant for cryptography, the CBOM contains information about the concrete cryptographic mechanisms in use. This information can be used in conjunction to plan a migration, e.\,g., by identifying vulnerable protocols in the communication between two assets in the CBOM and then using the SBOMs of these assets to locate the software that should be migrated.

Building an SBOM or a CBOM requires extensive tooling for \mbox{(semi-)}automated collection. There is a plethora of relevant tools available\footnote{see for example \url{https://perma.cc/N8WE-QBTZ}} that help either on the system level or on the network level to collect the desired information of the infrastructure. The results are typically vast: even for small setups, such an automated assembly can easily yield tens of thousands of components.

However, complex migration projects require many successive steps and in-depth strategic planning. 
Based on an SBOM/CBOM alone, only a simple risk-assessment and migration of components\,--\,without formal structure\,--\,based only on the highest risk exposure or when updates are available is feasible. We believe this is insufficient since it ignores the complexity of migration projects that should be based on more information than the SBOM and CBOM provide, including information about exposure, and should be based on a structured plan with a comprehensive overview rather than on ad-hoc interventions.

Initial standardization attempts with the same aim include frameworks like CARAF, a risk management approach for post-quantum migration~\cite{ma2021caraf}. In this framework, risks are identified, an asset inventory is set up and the risk is estimated for each asset, yielding an organizational road-map. Relating that to our migration model, after the asset inventorization we also get a list of relevant components (i.\,e., the nodes within the graph), but the CARAF approach does not provide any information about the dependencies between the components.

The draft ``Repeatable Framework for Quantum-safe Migrations'' of the framework ETSI published~\cite{etsi24cyberqsc} approaches the migration problem from an organizational perspective. It assumes that the responsibility for migration projects is assigned internally while some form of governance regulates the interdependence of the parties involved. Local dependency structures are described in part, but an analysis of longer chains of dependencies is not conducted. Therefore, only local dependencies can be captured by this framework.

As far as we know, a formal treatment for the migration problem is still pending: While Hassim et al.~\cite{hassim24framework} introduce a solid starting point for a PQC migration model by analyzing certain local dependencies in real-world migration projects, these authors do not work out the resulting composed dependency structures found in real-world interconnected systems. Instead they provide a practical framework focused on cryptographic and data inventories and a local dependency analysis within enterprise systems. Also, their analysis primarily addresses operational aspects.

Building on this groundwork, our efforts enhance the theoretical rigor by formalizing migration challenges through mathematical models. Our generalization even addresses broader migration scenarios beyond PQC-specific contexts.
We believe that our migration model together with its analysis and study of real-world applicability can substantially contribute to the understanding of the formal foundations of cryptographic migration.

In any case, combining different kinds of information, carving out relations and creating a model of the migration graph is still a lot of manual work and is currently not fully supported by tooling. In future work we envision tooling to support the creation and maintenance of migration graphs as well as their analysis.

%

%
%


\section{A Model for the Cryptographic Migration Problem}
\label{sec:formal_model}

Having motivated the need for a rigorous treatment of cryptographic migrations, we formulate a migration project as a directed graph whose vertices represent concrete software, hardware or organizational components which need to be migrated. Its edges capture dependencies that constrain the order of those migrations. We introduce an intuitive model, moving to a precise definition that introduces the key notions of strongly-connected migration clusters, migratable components, and the scheduling strategies that remove such clusters step-by-step. This formal framework sets the stage for the complexity analysis in Section~\ref{sec:migration_hard} and underlies the practical patterns developed later in the paper.

\subsection{Informal Description of the Model}

Any IT infrastructure we migrate consists of a number of \emph{migration components} or simply \emph{components} which need to be migrated. For an overview of which components we might encounter in practice see Section~\ref{sec:real-world-patterns}. The \emph{migration problem} is now to migrate all migration components. However, not every component can be migrated at every point in time, because there might be dependencies between the components. For instance, an application can migrate to post-quantum algorithms only after the cryptographic libraries it relies on have been upgraded. Another example is the migration of a communicating software component, which can only be migrated if the corresponding counterparts of the communication are migrated at the same time or if we introduce some kind of interoperability. In this case the migration can be done (if there are no other dependencies) if all communicating parties of this single migration component are migrated at the same time. We call this a \emph{migration cluster}. Components without dependencies outside its migration cluster are called \emph{migratable components}.
The \emph{migration process} looks as follows:
\begin{enumerate}
\item In every migration step we select a migratable component, following a suitable \emph{migration strategy}.
\item We then compute its migration cluster, i.\,e., possible interdependent components that need to be migrated at the same point in time.
\item Then, the actual migration of the migration cluster is performed and the successfully migrated migration cluster is removed from our model (together with all in-going dependencies).
\item This process is repeated until there is no migration component left and we have then solved the migration problem for this specific IT infrastructure.
\end{enumerate}
The whole model is now captured in a directed graph whose vertices are the migration
components and the edges are the dependencies between components, the \emph{migration graph} of the migration problem, which we will detail in this section.
Note that a migration graph's structure will vary significantly depending on the concrete IT infrastructure and the kind of model we have in mind. But also the abstraction layer and the type of migration lead to many design choices of the model, cf.~Section~\ref{sec:migration-modelling}.
The resulting graphs will show parallels to those that already exist in several disciplines of computer science such as flowcharts, call graphs, and finite-state machines in programming, network diagrams, or software dependency graphs.

We expect that post-quantum migration is the most significant type of cryptographic migration of IT infrastructures in the near future\ifanonymous\else~\cite{naether2024migrating}\fi. Its specific properties introduce aspects that were less prevalent in the classical cryptography world, where examples are hybrid cryptographic schemes and more complex instantiations of cryptographically agile solutions.
%
In a graph which only models systems and their protocol-related communication relations in between, one may only treat this implicitly. Yet again, while graphs that take other dependencies into account usually become harder to model, they sometimes better reflect the complicated circumstances of the real world.

Also, regulatory requirements of post-quantum solutions and different views on current and future algorithm recommendations may become significantly more complicated, especially in a multinational context.
All in all, post-quantum migration is the scenario that motivated our formal migration model. At the same time, it implicates several special properties which make modeling even more complex as we demonstrate in Section~\ref{sec:real-world-patterns}.

We now specify the formal model. The concepts concerning directed graphs are partially quite standard, nevertheless we specify for the sake of completeness the full notation. Novel and unprecedented in the literature is the application to migration problems. The reader interested in practical applicability might want to skip the following subsections and the complexity proofs of Section~\ref{sec:migration_hard}, and continue reading in Section~\ref{sec:migration-modelling}.

\subsection{Formal Definition of Migration Graphs}
\label{sec:graph-model}

Consider a finite set $V$ of components in a large IT infrastructure. If the migration of a component $v \in V$ depends on the migration of a component $w \in V$, we write $v \to w$. In other words: We cannot migrate $v$ before migrating $w$ (but possibly at the same time). Collecting all such dependencies in a set $E$ (without loops $v \to v$) gives us the \emph{migration graph} $G = (V,E)$ modeling the \emph{migration project} we have in mind. Canonically, we extend the relation $v \to w$ to sets $M_1, M_2 \subseteq V$ and write $M_1 \to_G M_2 := \Set{ v \to w ; v \in M_1, w \in M_2} \subseteq E$. If $G$ is clear from the context, we omit the subscript and only write $M_1 \to M_2$. For a dependency $e = v \to w \in E$ we write $~source(e) := v$ and $~target(e) := w$, again with canonical extension to sets. For later use, we write for a set $M \subseteq V$ for the \emph{restricted graph} $G|_M = (M, M \to M)$. The \emph{direct dependencies} of $v \in V$ are precisely given by the set $~target(v \to V)$.
We define the \emph{dependencies of $v$} as $~dep(v) :=~target(v \to^* V)$, where $\to^*$ is the reflexive and transitive closure of the relation $\to$. 
%
A simple example of a migration graph can be found in Fig.~\ref{fig:example_graph}.
\begin{figure*}[tb]
\centering

\begin{tikzpicture}[auto,
            > = Stealth,
every edge quotes/.style = {font=\footnotesize}, 
every edge/.append style = {->, draw=bluelight, thick},
every loop/.append style = {<-, looseness = 12},
node distance = 15mm,
 state/.style = {circle, semithick, fill=bluelight, draw=bluelight, text=white, minimum size=2.5em},
 initial text = ,]

\node (L) [state] {$v_{10}$};
\node (K) [state, below=.9cm of L] {$v_{11}$};
\node (J) [state, below=.9cm of K] {$v_{12}$};
\node (I) [state, right=of K] {$v_9$};
\node (H) [state, below right=of I] {$v_8$};
\node (G) [state, above right=of I] {$v_7$};
\node (F) [state, right=of H] {$v_6$};
\node (E) [state, right=of G] {$v_5$};
\node (D) [state, below right=of E] {$v_4$};
\node (B) [state, right=of D] {$v_2$};
\node (C) [state, below=.9cm of B] {$v_3$};
\node (A) [state, above=.9cm of B] {$v_1$};

\path   (J) edge (I)
	(K) edge (I)
	(L) edge (I)
        (H) edge (I)
        (G) edge (H)
        (I) edge (G)
        (H) edge (F)
        (G) edge (E)
        (F) edge [bend left=10] (E)
        (E) edge [bend left=10] (F)        
        (F) edge (D)
        (E) edge (D)
        (D) edge (C)
        (D) edge (B)
        (D) edge (A)
        ;

\end{tikzpicture}

\caption{An example of a simple migration graph on twelve components $v_1, \dots, v_{12}$.\label{fig:example_graph}}
\end{figure*}

\subsection{Migration Clusters and Migratable Components}

\begin{definition}[Migration Cluster]
\label{def:migration-cluster}
Let $G = (V,E)$ be a migration graph and $v \in V$. The \emph{migration cluster $c(v)$ of $v$} is set of all components $w \in V$ such that $w \in~dep(v)$ and $v \in~dep(w)$.
\end{definition}
In the literature, the migration clusters are classically called the \emph{strongly connected components} of $G$.
For sake of exposition we stick to the term \emph{migration cluster} in the context of cryptographic migration.
Intuitively, the elements of $c(v)$ are precisely those components that have to be migrated at the same time as $v$ due to mutual dependencies. The migration clusters specify how complex single \emph{migration steps} of a \emph{migration project} are. Additionally, the properties of the migration clusters inside a migration graph $G$ shed light on the complexity of the resulting \emph{migration problem} which we will describe below.
In fact, we have the following important characteristics:
\begin{enumerate}
\item The \emph{size} $s$ of the cluster $s(v) := \# c(v)$.
\item The \emph{degree} $~cdeg(v) := ~argmax_{w \in G |_{c(v)}}~deg(w)$.
\end{enumerate}
Clearly, in the case $~cdeg(v) = 0$ we have $s(v) = 1$ since otherwise there would be two components $v_1 \neq v_2 \in c(v)$ with $v_1 \in~dep(v_2)$ and $v_2 \in~dep(v_1)$. But then $~dep(v_1) =~target(v_1 \to^* c(v_2))$ would be non-empty and thus $v_1 \to^* c(v_2)$ would also be non-empty. Therefore, we have $~deg_{G|_{c(v_1)}}(v_2) > 0$, contradicting our assumption. Conversely, if $s(v) = 1$ then $~cdeg(v) = 0$ since $G|_{c(v)}$ consists of a single component only and has no edges. In practice this case occurs in the situation where we have a single component that simply can be migrated without any dependencies on other components. This is the easiest case in a migration process.

In the case $~cdeg(v) = 1$ the migration cluster $c(v)$ is a cycle. In practice this case often occurs for $s(v) = 2$ when there are two components that communicate using the same protocol. A migration of the protocol can only be done successfully if both components are migrated at the same time. Cycles of larger sizes are uncommon in practice, cf.~Section~\ref{sec:real-world-patterns}.
The set $m(v)$ of \emph{migratable components around $v$} is given by
\[
m(v) = \begin{cases}
c(v) & \text{, if $c(v) \to \bar{c}(v) = \emptyset$}\\
\emptyset & \text{, otherwise}\\
\end{cases}
\]
where we write $\bar{c}(v) = V \setminus c(v)$. The set $m(v)$ specifies exactly which components can be migrated in $G$ when we migrate $v \in V$. If $v$ has dependencies outside its migration cluster, i.\,e., $v$ cannot be migrated, then $m(v) = \emptyset$. If there are no direct dependencies of $v$, then $m(v) = \Set{v}$. Also, in all other cases $m(v) = c(v)$, i.\,e., the whole migration cluster. If $m(v) \neq \emptyset$, we call $v$ \emph{migratable} (in $G$).

Now, in a migration project, there might be several subsequent steps to fully migrate all involved components, i.\,e., solve the migration problem the migration graph poses. We write $G^{(v)} = (V \setminus m(v), E \setminus (V \to m(v)))$ for the graph after a migration attempt of $v$. If $v$ is not migratable, then $G^{(v)} = G$. Generally, $G^{(v)}$ is exactly the graph  restricted to the components outside the migratable components around $v$. Using similar to above the notation $\bar{m}(v) = V \setminus m(v)$, we can thus alternatively write $G^{(v)} = G|_{\bar{m}(v)} = (\bar{m}(v), \bar{m}(v) \to \bar{m}(v))$. We now show
\begin{lemma}
\label{lemma:migratable}
In any (finite and non-empty) migration graph $G = (V,E)$ there is at least one component $v \in V$ that is migratable.
\end{lemma}
\begin{proof}
Suppose there is no migratable component, i.\,e., $m(v) = \emptyset$ for all $v \in V$. By definition, $c(v) \to \bar{c}(v)$ is always non-empty, i.\,e., there is always a dependency to a component outside the migration cluster of $v$. Call such a dependency $w \in \bar{c}(v)$. Since $w$ is outside the migration cluster of $v$, $w$ cannot depend on elements of $c(v)$ and all dependencies of $w$ reside in the (strict) sub-graph $G|_{\bar{c}(v)}$. Repeating this process, after a finite number of steps, we end-up in a graph which only contains a single migration cluster $c(u)$ and $u \in V$ is migratable in $G$ contradicting our assumption.
\end{proof}


\subsection{Migration Strategies}
\label{sec:migration_strategies}

We will now define \emph{migration strategies of $G$}. Formally, a strategy $S = (v_1, \ldots, v_\ell) \in V^\ell$ for $\ell \in \mathbb{N}$ is a sequence of components that are migrated one by one (of course together with their respective migration clusters). $\ell$ is the \emph{length} of the strategy. For $1 \leq i \leq \ell$ denote by $G_i$ the graph after successively migrating $v_1, \ldots, v_i$, i.\,e., $G_i := G_{i-1}^{(v_i)}$, where $G_0 = G$. The strategy $S$ is \emph{reasonable} if for all $1 \leq i \leq \ell$, we have $m(v_i) \neq \emptyset$ in the graph $G_{i-1}$. That means that the component $v_i$ is in fact migratable in the graph $G_{i-1}$ and, therefore $G_{i} \subsetneq G_{i-1}$, i.\,e., the migration of $v_i$ yields a smaller graph. $S$ is \emph{successful}, if $G_{\ell}$ is the empty graph. By Lemma~\ref{lemma:migratable} there is a reasonable, successful strategy for any migration graph $G$.

\begin{figure*}[t]
\centering

\begin{tikzpicture}[auto,
            > = Stealth,
every edge quotes/.style = {font=\footnotesize}, 
every edge/.append style = {->, draw=bluelight, thick},
every loop/.append style = {<-, looseness = 12},
node distance = 15mm,
 state/.style = {circle, semithick, fill=bluelight, draw=bluelight, text=white, minimum size=2.5em},
 initial text = ,]

\node (L) [state] {$c_{7}$};
\node (K) [state, below=.9cm of L] {$c_{8}$};
\node (J) [state, below=.9cm of K] {$c_{9}$};
\node (I) [state, right=of K] {$c_6$};
\node (E) [state, right=of I] {$c_5$};
\node (D) [state, right=of E] {$c_4$};
\node (B) [state, right=of D] {$c_2$};
\node (C) [state, below=.9cm of B] {$c_3$};
\node (A) [state, above=.9cm of B] {$c_1$};

\path   (J) edge (I)
	(K) edge (I)
	(L) edge (I)
        (I) edge (E)
        (E) edge (D)
        (D) edge (C)
        (D) edge (B)
        (D) edge (A)
        ;

\end{tikzpicture}

\caption{The acyclic condensation graph $G/c$ on nine clusters $c_1, \dots, c_9$ of the migration graph from Fig.~\ref{fig:example_graph}.\label{fig:contraction}}
\end{figure*}

\begin{lemma}
\label{lemma:migration_length}
All reasonable, successful strategies $S = (v_1, \ldots, v_\ell)$ for a migration graph $G = (V,E)$ have the same length $\ell$.
\end{lemma}
\begin{proof}
Consider two reasonable, successful strategies $S_1 = (v_1^{1}, \ldots, v_{\ell_1}^{1})$ and $S_2 = (v_1^{2}, \ldots, v_{\ell_2}^{2})$ with $\ell_2 \geq \ell_1$. Since $S_1$ is successful, $V = \biguplus_{1 \leq i \leq \ell_1} m(v_i)$ partitions into the migration clusters given by $S_1$. Note that all $m(v_i) = c(v_i)$ are non-empty since $S_1$ is reasonable. Now, consider the (sub-)strategy $S_3 = (v_1^{2}, \ldots, v_{\ell_1}^{2})$ of the strategy $S_2$. By the partition property, there is for each $1 \leq j \leq \ell_1$ an $1 \leq i \leq \ell_1$ such that $v_j^{2} \in c(v_i^{1})$. Furthermore, for each $1 \leq j_1 < j_2 \leq \ell_1$ and corresponding $1 \leq i_1, i_2 \leq \ell_1$, we have $i_1 \neq i_2$ since otherwise $v_{j_2}^{2} \in m(v_{j_1}^{2}) \neq \emptyset$. The non-emptiness of $m(v_{j_1}^{2})$ follows from the assumption that $S_2$ is reasonable. But then $\sum_{1 \leq j \leq \ell_1} s(v_{j}^{2}) = \#V$ and the strategy $S_3$ is already successful.
\end{proof}
Due to Lemma~\ref{lemma:migration_length} the \emph{migration length} $\ell(G)$ as a property of the migration graph $G$ is well-defined. It basically expresses how many independent migration steps one has to perform in order to fully migrate a system. Equivalently, we see that the set of components $V$ partitions into the migration clusters of its components, i.\,e., $V = \bigcup_{v \in V} c(v)$ and for all $v_1, v_2 \in V$, we either have $c(v_1) = c(v_2)$ or $c(v_1) \cap c(v_2) = \emptyset$. In other words the relation $v_1 \sim_c v_2 :\Leftrightarrow c(v_1) = c(v_2)$ is an equivalence relation. 
%
Another important measure of the complexity of a migration project is the
\emph{migration depth} of the migration graph $G$. Clearly, migratable independent
components, i.\,e., migratable $v_1, v_2 \in V$ with $c(v_1) \neq c(v_2)$ can in practice either be migrated in parallel or sequentially.
As it turns out there is a \emph{canonical strategy} $S(G)$ (up to isomorphism) that specifies how in practice a migration project should be performed optimally. The idea is to specify a directed acyclic graph $G/c = (V/c, E/c)$, classically called the \emph{condensation of $G$}, containing as nodes the migration clusters $c(v)$ that is $V/c = \Set{ c(v) ; v \in V}$ with an edge $c(v_1) \to c(v_2)$ iff there is an edge between a member of $c(v_1)$ and a member of $c(v_2)$, i.\,e., $E/c = \Set{ c(v_1) \to c(v_2) ; v_1 \to v_2 \in E }$, see Fig.~\ref{fig:contraction}. We now have $\#(V/c) = \ell(G)$. We state the simple and well-known lemma whose proof is straightforward:
\begin{lemma}
\label{thm:acyclic_condensation}
For a migration graph $G = (V,E)$ the graph $G/c$ is acyclic.\qed
\end{lemma}

\label{def:canonical_slice}
Notions like migratability can now also be considered for the graph $G/c$. We define the \emph{canonical migration time} of $c(v) \in V/c$ as follows: $t(c(v)) :=~min_{1\leq k \leq \ell(G)} c(v) \to^k V/c = \emptyset$. The canonical migration time is the position in the topological order of the directed acyclic graph $G/c$. Migratable migration clusters have $t(c(v)) = 1$, all other clusters have $t(c(v)) > 1$. Write $T(k) := \Set{ c(v) \in V/c ; t(c(v)) = k }$ for the set of all migration clusters with canonical migration time $k$. The \emph{canonical strategy} is now to first migrate all migration clusters in $T(1)$ then in $T(2)$ and so on. A canonical strategy is thus given by $S(G) = (T(1), \ldots, T(d(G)))$. We call $d(G) :=~argmax_{c(v) \in V/c} t(c(v))$ the \emph{migration depth} of $G$. In other words, the migration depth specifies how many migration steps successively depend on each other in the worst case. Note that for all $1 \leq i \leq d(G)$ we can arbitrarily order the elements of the set $T(i)$. Thus, there are exactly $\prod_{1 \leq i \leq d(G)} \#T(i)!$ many canonical strategies. Since all of them are isomorphic (w.r.t.\ the dependency relation on $G/c$), our notation for $S(G)$ is sound.

The overall structure of the migration graph $G$ is thus as follows (cf.~Fig.~\ref{fig:clusters}):

\begin{enumerate}
\item $G$ partitions into the migration clusters of its components.
\item The migration clusters have only acyclic dependencies between each other.
\item The migration clusters have thus a canonical position in the migration process depending on how deep in the graph they are located.
\end{enumerate}

Note that there are also non-canonical migration strategies, where a migration of a component $v \in V$ removed a dependency of another component $w \in V$ and even though there were other migratable components in the first place the migration of $w$ is done next.
%

\begin{figure*}[t]
\centering

\begin{tikzpicture}[auto,
            > = Stealth,
every edge quotes/.style = {font=\footnotesize}, 
every edge/.append style = {->, draw=bluelight, thick},
every loop/.append style = {<-, looseness = 12},
node distance = 15mm,
 state/.style = {circle, semithick, fill=bluelight, draw=bluelight, text=white, minimum size=2.5em},
 initial text = ,]

\node (L) [state] {$v_{10}$};
\node (K) [state, below=.9cm of L] {$v_{11}$};
\node (J) [state, below=.9cm of K] {$v_{12}$};
\node (I) [state, right=of K] {$v_9$};
\node (H) [state, below right=of I] {$v_8$};
\node (G) [state, above right=of I] {$v_7$};
\node (F) [state, right=of H] {$v_6$};
\node (E) [state, right=of G] {$v_5$};
\node (D) [state, below right=of E] {$v_4$};
\node (B) [state, right=of D] {$v_2$};
\node (C) [state, below=.9cm of B] {$v_3$};
\node (A) [state, above=.9cm of B] {$v_1$};

\path   (J) edge (I)
	(K) edge (I)
	(L) edge (I)
        (H) edge (I)
        (G) edge (H)
        (I) edge (G)
        (H) edge (F)
        (G) edge (E)
        (F) edge [bend left=10] (E)
        (E) edge [bend left=10] (F)
        (F) edge (D)
        (E) edge (D)
        (D) edge (C)
        (D) edge (B)
        (D) edge (A)
        ;

        \coordinate (mid41) at ($(J)!0.5!(I)$);
    	\coordinate (mid42) at ($(L)!0.5!(I)$);
        
        \coordinate (mid31) at ($(G)!0.5!(E)$);
    	\coordinate (mid32) at ($(H)!0.5!(F)$);
        
        \coordinate (mid21) at ($(E)!0.5!(D)$);
    	\coordinate (mid22) at ($(F)!0.5!(D)$);
        
        \coordinate (mid11) at ($(D)!0.5!(A)$);
    	\coordinate (mid12) at ($(D)!0.5!(C)$);
        
	\node[draw, dashed, draw=bluelight, text=white, rounded corners, fit=(A), inner sep=2mm, label=below:$c_1$] {};
	\node[draw, dashed, draw=bluelight, text=white, rounded corners, fit=(B), inner sep=2mm, label=below:$c_2$] {};
	\node[draw, dashed, draw=bluelight, text=white, rounded corners, fit=(C), inner sep=2mm, label=below:$c_3$] {};
	
	\node[draw, dashed, draw=bluelight, text=white, rounded corners, fit=(D), inner sep=2mm, label=below:$c_4$] {};
        \node[draw, dashed, draw=bluelight, text=white, rounded corners, fit=(E) (F), inner sep=2mm, label=below:$c_5$] {};
        
        \node[draw, dashed, draw=bluelight, text=white, rounded corners, fit=(G) (H) (I), inner sep=2mm, label=below:$c_6$] {};
        
	\node[draw, dashed, draw=bluelight, text=white, rounded corners, fit=(J), inner sep=2mm, label=below:$c_9$] {};
	\node[draw, dashed, draw=bluelight, text=white, rounded corners, fit=(K), inner sep=2mm, label=below:$c_8$] {};
	\node[draw, dashed, draw=bluelight, text=white, rounded corners, fit=(L), inner sep=2mm, label=below:$c_7$] {};

	\coordinate (top) at ([yshift=.4cm]current bounding box.north);
	\coordinate (bottom) at ([yshift=-.8cm]current bounding box.south);

	\draw[dashed, draw=black!40] (mid41 |- top) -- (mid41 |- bottom);
	\draw[dashed, draw=black!40] (mid31 |- top) -- (mid31 |- bottom);
	\draw[dashed, draw=black!40] (mid21 |- top) -- (mid21 |- bottom);
	\draw[dashed, draw=black!40] (mid11 |- top) -- (mid11 |- bottom);

	\node[fit=(A) (B) (C), inner sep=3mm,minimum height=6cm, label=below:$T(1)$] {};
	\node[fit=(D), inner sep=3mm,minimum height=6cm,label=below:$T(2)$] {};
        \node[fit=(E) (F), inner sep=3mm,minimum height=6cm,label=below:$T(3)$] {};
        \node[fit=(G) (H) (I), inner sep=3mm,minimum height=6cm,label=below:$T(4)$] {};
        \node[fit=(J) (K) (L), inner sep=3mm,minimum height=6cm,label=below:$T(5)$] {};

\end{tikzpicture}

\caption{The migration graph from Fig.~\ref{fig:example_graph} with its migration clusters (denoted by dashed boxes) and canonical migration times. In this example with $12$ components, we have $\ell(G) = 9$ and $d(G) = 5$. \label{fig:clusters}}
\end{figure*}

\section{Assessing the Difficulty of the Migration Problem}
\label{sec:migration_hard}

We now show that, as expected, any sufficiently large migration project splits into many dependent migration steps that cannot be parallelized and provide a formal proof of this statement.

From classical results in combinatorial analysis we know the following about random decompositions of the migration graph into its clusters.

\subsection{Random Migration Clusters}

Since every migration graph decomposes disjointly into its migration clusters, we now study the expected properties of these migration clusters. Let $n$ be the number of components. There are $B_n$ different possible partitions of $n$ components in disjoint subsets, where 
\[
B_n = \sum_{k=0}^n {n \choose k} B_k = \frac{1}{e} \sum_{k=0}^\infty \frac{k^n}{k!}
\]
is the $n$th Bell number. More precisely, there are $n \brace k$ partitions of our components in $k$ distinct non-empty subsets, where $n \brace k$ are called the \emph{Stirling numbers of the second kind}. Obviously, we have $B_n = \sum_{1 \leq k \leq n} {n \brace k}$.

\subsubsection{Sampling Migration Clusters}
\label{sec:set_partitions}

Consider now a migration project with $m$ different sub-projects and suppose each of our $n$ components is assigned to one of the $m$ sub-projects uniformly at random. Let $\Pi_{nm}$ be the random variable defining the resulting partition $\pi$ of our components. We have:
\[
~prob(\Pi_{nm} = \pi) = \frac{m^{\underline{k}}}{m^n}
\]
if $\#\pi = k$, where $m^{\underline{k}} = m \cdot (m-1) \cdots (m-k+1)$ is the falling factorial function. In other words: if there are $k$ different sub-projects with at least one component, then the probability of creating one such partition is given by the formula above. For the distribution of the number of non-empty different sub-projects we have
\[
~prob(\#\pi = k) = {n \brace k}\frac{m^{\underline{k}}}{m^n}.
\]
Note that for arbitrary migration graphs, we do not know the distribution of the different partitions. We have the following
\begin{theorem}[Proposition 2 in~\cite{pit97b}]
\label{thm:partition_sampling}
Let $M$ be the random variable with values in $\N_{\geq 1}$ and suppose given $M = m$ that $n$ balls labeled by $\N_{\leq n}$ are thrown independently and uniformly at random into $m$ boxes. Then the following two statements are equivalent:
\begin{enumerate}
\item $M$ is distributed via $~prob(M = m) = m^n p_m / B_n$ for $m \in \N_{\geq 1}$, where $p_m$ is a probability distribution on $\N_{\geq 0}$ whose first $n$ factorial moments are identically equal to $1$. One may for example take $p_m = 1/(e m!)$.
\item $\Pi_{nM}$ has uniform distribution over the set of all partitions of $n$ components.
\end{enumerate}
\end{theorem}
The Theorem gives a constructive way of uniformly generating a migration cluster: First sample $m \in \N_{\geq 1}$ following the above distribution, then distribute the $n$ components uniformly in $m$ bins. The resulting partition $\pi$ will be a uniformly selected partition of our $n$ components. 
We will now study the expected properties of the outcome.

\subsubsection{Distribution of Random Migration Clusters}
\label{sec:random_clusters}

Studying the distribution of random partitions of sets with $n$ elements is a classical problem. Set into the context of migration problems, we state
\begin{theorem}[Summary of~\cite{odlric85} and Chapter 4, Theorem 1.1 in~\cite{sac97}]
\label{thm:cluster-distribution}
The number of migration clusters $M$ is expectedly $\mathsf{E}(M) \sim n/~log n$ and the expected size of each migration cluster is $~log n$. Also, almost all partitions have approximately $e ~log n$ distinct cluster sizes. More specifically the distribution of the number of migration clusters $M$ converges asymptotically to a normal distribution with mean $\mu = n/~log n$ and standard deviation $\sigma = \sqrt{n}/~log n$ for $n \to \infty$.
\end{theorem}
Furthermore, the expected maximum size of a migration cluster is, for growing $n$, asymptotic to $e r - ~log \sqrt{r}$, where $r \cdot ~exp(r) = n$ and $r = ~log n - (1+o(1)) ~log ~log n \sim ~log n$~\cite{pit97}. This means that we typically expect lots of small clusters and we do not expect exceptionally large ones. Results on the speed of convergence are beyond the scope of this article. We now compute an upper bound on the expected number of migration clusters when looking independently at $s$ randomly selected migration projects.
\begin{theorem}
\label{thm:cluster-distribution-max}
The expected maximal number of migration clusters when looking independently at $s$ randomly selected migration projects is bounded above by
\[
\frac{n}{~log n} + \frac{\sqrt{2 n ~log s}}{~log n}.
\]
\end{theorem}
\begin{proof}
Let $M_1, \ldots, M_s$ be the resulting numbers of migration clusters and define $X = ~max_{1 \leq i \leq s} M_i - \mu$. Then conveniently $\mathsf{E}(X) = 0$.
Since for a fixed real constant $t$ the function $x \mapsto ~exp(t x)$ is convex, by Jensen's inequality we have:
\begin{align*}
~exp(t \mathsf{E}(X)) & \leq \mathsf{E}(~exp(t X)) = \mathsf{E}(\max_{ 1\leq i\leq s}(~exp(t M_i)))\\
& \leq \sum_{i = 1}^{s} \mathsf{E}(~exp(t M_i))
 = s ~exp(t^2 \sigma^2 / 2)
\end{align*}
Here, the last equality follows from the fact that the moment generating function $\mathsf{E}(~exp(t M_i))$ for a normal distribution with mean zero is equal to $~exp(t^2 \sigma^2 / 2)$ and in particular independent of the index-variable $i$.
Rewriting this, $\mathsf{E}(X) \leq \frac{~log s}{t} + \frac{t \sigma^2}{2}$. This bound is minimal at $\frac{\sqrt{2 ~log s}}{\sigma}$ since its derivative in direction $t$ is $\frac{\sigma^2}{2} - \frac{~log s}{t^2}$. Thus, $\mathsf{E}(X) \leq \sigma \sqrt{2 ~log s}$. 
\end{proof}
In practice, the Lemma implies that we would have to sample exponentially many migration graphs (in the number of components $n$) in order to expectedly observe graphs with asymptotically more than $n / ~log n$ migration clusters. Thus, the statement from Theorem~\ref{thm:cluster-distribution} applies in practice all the time, at least for sufficiently large $n$.
%
It remains to study the dependencies between the migration clusters, i.\,e., the structure of the condensation modulo the migration clusters.

\begin{table*}[t]
\caption{Numerical evaluation of the expected behavior of migration clusters for migration projects with $n$ components rounded to the nearest integer. We see that one does not expect large mutual dependencies but the expected number of independent clusters is vast and the chain-dependencies are also large when we bound the number of migration clusters with fixed canonical migration time by $w(n) = ~log n$ or $w(n) = \sqrt{n} / ~log n$, respectively.}
\label{tab:numerical_behavior}
\[
\begin{array}{c@{\qquad}c@{\qquad}c@{\qquad}c@{\qquad}|@{\qquad}c@{\qquad}c}
\toprule
\#V & \mathsf{E}(\#c(v)) & \mathsf{E}(\#(G/c)) & \mathsf{SD}(\#(G/c)) & \multicolumn{2}{c}{d(G)}\\
n   & ~log n & \frac{n}{~log n} & \frac{\sqrt{n}}{~log n} & \frac{n}{~log^2 n} & \sqrt{n}\\
\midrule
\spacednumber{10} & \spacednumber{2} & \spacednumber{4} & \spacednumber{1} & \spacednumber{2} & \spacednumber{3}\\
\spacednumber{100} & \spacednumber{5} & \spacednumber{22} & \spacednumber{2} & \spacednumber{5} & \spacednumber{10}\\
\spacednumber{1000} & \spacednumber{7} & \spacednumber{145} & \spacednumber{5} & \spacednumber{21} & \spacednumber{32}\\
\spacednumber{10000} & \spacednumber{9} & \spacednumber{1086} & \spacednumber{11} & \spacednumber{118} & \spacednumber{100}\\
\spacednumber{100000} & \spacednumber{12} & \spacednumber{8686} & \spacednumber{27} & \spacednumber{754} & \spacednumber{316}\\
\spacednumber{1000000} & \spacednumber{14} & \spacednumber{72382} & \spacednumber{72} & \spacednumber{5239} & \spacednumber{1000}\\
\spacednumber{10000000} & \spacednumber{16} & \spacednumber{620421} & \spacednumber{196} & \spacednumber{38492} & \spacednumber{3162}\\
\spacednumber{100000000} & \spacednumber{18} & \spacednumber{5428681} & \spacednumber{543} & \spacednumber{294706} & \spacednumber{10000}\\
\spacednumber{1000000000} & \spacednumber{21} & \spacednumber{48254942} & \spacednumber{1526} & \spacednumber{2328539} & \spacednumber{31623}\\
\bottomrule
\end{array}
\]
\end{table*}

\subsection{Properties of the Condensation Graph}
\label{sec:migration_hard:condensation_graph}
According to Theorem~\ref{thm:cluster-distribution}, we consider a condensation graph, i.\,e., the directed acyclic graph resulting from the migration graph on $n$ components modulo its strongly connected components, with expectedly $n' = n / ~log n$ vertices. Since we are considering a single migration project, we can assume that this graph is connected.
How can we model this graph?

A first idea would be to use the standard Erdős–Rényi model~\cite{erdren1960evolution}, where the graph is constructed by connecting the components independently at random with a fixed probability $p$. One can ensure that the result is acyclic by defining any linear order on the vertices and include a directed edge only if it obeys the selected ordering. Each such edge is included in the graph with probability $p$, independently from every other edge.
For such graphs, we have
\begin{theorem}[Theorem 4.1 in~\cite{frie23-random-graphs}, a variant of Theorem 1 in~\cite{erdren59random}]
\label{thm:connected}
Let $G_{n, p_n}$ be a random $n$-vertex graph, where each arc is present with probability $p_n = \frac{~log n + c_n}{n}$. Then, for growing $n$, the probability that $G$ is connected is given by
\[
\lim_{n \to \infty} ~prob(G_{n, p_n} \text{ connected}) = \begin{cases}
0 & \text{, if $c_n \to -\infty$}\\
e^{-e^{-c}} & \text{, if $c_n \to c$ const.}\\
1 & \text{, if $c_n \to \infty$}\\
\end{cases}
\]
\end{theorem}
Note that we can directly apply this Theorem to our directed acyclic graph situation, since such a graph is connected if and only if the undirected variant of it is connected.

As our goal is to prove a lower bound on the number of edges in the condensation graph (with $n'$ vertices!) in our model, we assume now the sparsest case $c_{n'} \to c$ constant for which the graph is connected with high probability. Note that the doubly-exponential term in Theorem~\ref{thm:connected} approaches $1$ already for comparatively small $c$ very quickly. For example, already for $c = 10$, the value of $e^{-e^{-c}} \geq 0.99995$. But also for any $c > 0$, we obtain for the number of edges in the condensation graph expectedly asymptotically $\frac{~log n' + c_{n'}}{n'} \cdot \frac{n' \cdot (n'-1)}{2} = \frac{1}{2}((n'-1) (~log n' + c_{n'})) \sim \frac{1}{2} n' ~log n'$.
That means that even in the sparsest case, we expect a considerable number of dependencies between the migration clusters.

We now analyze what happens if we bound the number of clusters with a common ancestor of distance $k$ (the maximum being the \emph{width $w$} of the graph), i.\,e., $\#(c(v) \to^k V/c)$ by a small function $w(n') = o(n')$ in the size $n'$ of the condensation graph $G/c$.

\begin{theorem}
\label{thm:large_depth}
Let $\shortmap[w]{\mathbb{R}}{\mathbb{R}}$ be a smooth function in a real variable with $w \in o(n')$, e.\,g., $w(n') = ~log n'$. Assume for any cluster $c(v)$ in a condensation graph $G/c$ with $n'$ components, we have for each $k \geq 0$ the bound $\# (c(v) \to^k V/c) \leq w(n')$. Then $d(G) \geq n'/w(n')$. In other words, the width of the condensation graph has a trade-off with the depth of the condensation graph.
\end{theorem}
\begin{proof}
Without loss of generality we can assume that $G/c$ has a single source $c(v)$. Otherwise take the set of sources and consider the one with the longest path in $c(v) \to^* V/c$. This path has length $d(G)$. For $0 \leq k < d(G)$ define $W_k := c(v) \to^k V/c$. Since $G/c$ has only a single source, $\biguplus_{0 \leq k < d(G)} W_k = V/c$, e.\,g., $n' = \#(V/c) = \#\biguplus_{0 \leq k < d(G)} W_k = \sum_{0 \leq k < d(G)} \#W_k$. By assumption $\#W_k \leq w(n')$ for all $0 \leq k < d(G)$. Thus $n' \leq d(G) \cdot w(n')$ and the claim follows.
\end{proof}

If, e.\,g., in the Theorem $w(n') = ~log n'$, we obtain expectedly $d(G) \geq \frac{n'}{~log n'} = \frac{n}{(~log n - ~log~log n) ~log n}$, that is, large chains of dependencies of the migration clusters. Of course one can trade depth $d(G)$ vs.\ the bound $w(n')$, but then the complexity of at least one migration step will increase.

\subsection{Discussion}

In the previous sections we analyzed how different graph structures affect the
properties of our migration problem. To do so, we looked at random migration graphs, arguing that a random migration project should exhibit random clustering and the remaining directed acyclic condensation should also have a random structure. Thereby, we found several mathematical arguments that underline the fact that migration is hard in general.
Specifically, for a random (large) migration graph with $n$ components, we showed a number of theorems which are combined into the following

\begin{corollary}
\label{thm:hard_migration}
In expectation (of the random variables involved), we have:
\begin{enumerate}
\item Partitions can be efficiently sampled (Theorem~\ref{thm:partition_sampling}), i.\,e., we can easily observe different random migration graphs.

\item The migration graph decomposes in $n /~log n$ migration clusters of size $~log n$ with standard deviation $\sqrt{n} /~log n$ (Theorem~\ref{thm:cluster-distribution}).

\item Even after observing many migration projects, we do not
expect to observe any with considerably more than $n /~log n$ migration clusters (Theorem~\ref{thm:cluster-distribution-max}).

\item The condensation graph is connected and thus not sparse (Theorem~\ref{thm:connected}).

\item If for each node the number of dependencies is suitably bounded, the condensation graph has large depth, i.\,e., many migration clusters that successively depend on each other and cannot be parallelized. Stated the other way around, we have at least one migration step with many dependencies if we assume small depth (Theorem~\ref{thm:large_depth}).

\end{enumerate}
\end{corollary}

\noindent We thus showed that informally we have
\begin{fact}
\label{fact:hard_migration}
Any sufficiently large migration project splits into many dependent migration steps that cannot be parallelized. The migration of these clusters either takes many successive steps or includes at least one particularly difficult one.
\end{fact}
The statement tells us that a ``random'' migration project exhibits a certain complexity. We revise the assumptions made: The first assumption is that in any practical migration project, partitions (i.\,e., clustering) of components appear uniformly from the set of all partitions. Another assumption is that the structure of the condensation graph follows the Erdős–Rényi model, i.\,e., that dependencies appear with a fixed probability independently of each other. Finally, our third assumption is that the two modeling steps are structurally independent. All assumptions seem valid for the generic complexity assessment of the model. In practice, of course, we do not have asymptotic objects or expected behaviors of concrete migration graph examples. Nevertheless, the mathematical model provides a sound basis for understanding the structure of migration projects in general.
We summarize these results with a numerical evaluation of the relevant estimates in Table~\ref{tab:numerical_behavior}. In the table, we compute the numerical outcomes for different migration projects with $n = 10^i$ components for $i \in \Set{1, \dots, 9}$, give the expected size of the migration clusters, the expected number of migration clusters with its standard deviation and numerical values of the expected depth of the graph for different bounds of the width $\#(c(v) \to^k V/c) \leq w(n)$.


\section{Migration Modelling}
\label{sec:migration-modelling}

We will now describe some practical considerations when modelling the migration problem using a migration graph.

\subsection{Graph Models}

Today's large scale IT infrastructures are both complex and highly diverse, and examining models to create a suitable graph from a
real-world IT infrastructure is the subject of ongoing research.
The aggregation of the
necessary data that is used to derive the graph is highly individual to the systems that should be migrated and specific to the given migration goal. We state a few concrete examples for such migration contexts:
\begin{description}
\item[Network migration:]
When thinking about migrating a network, it is intuitive to think of
the components to represent computers, servers, and similar devices and their dependencies. The edges would then, e.\,g., model communication paths for specific (cryptographic) protocols. Depending on the migration goals, all connections between machines may be interesting regardless of the specific protocol, illustrating all the relations between machines in the network. Depending
on the objective, it can also be necessary to study several graphs\,--\,for example
individual graphs for protocols like TLS/HTTPS, SSH, or IPsec/IKEv2.
In general, we expect that the migration goal is typically to replace outdated cryptographic protocols with more modern versions. We elaborate on this in the next section and study a concrete example in Section~\ref{sec:real-world-examples:gitlab}.
\item[System migration:]
Different graphs may be modeling a specific system. For instance, the migration components could be a set of software artifacts within an operating system and the edges would model the dependencies between these artifacts, similar to SBOMs. We expect that the goal in such a case would be to identify those software components that need to be updated to support current cryptographic standards, e.\,g., by switching out libraries for newer version or recompiling some software with updated source code to address vulnerabilities. 
\item[Full-scale migration:]
For a system administrator, a tool focusing on the high-level network view may be enough, while programmers may profit from call-graph-like models. A full view on all abstraction levels may be provided by a combination of different graphs, e.\,g., a component being linked to another graph depicting the same component in a different context. This leads to several graph models of the same infrastructure each with its specific perspective. On the abstract level, an administrator might identify a node in the network that uses an outdated protocol using the ``network migration'' view. On the detailed level, the administrator could then identify the concrete software that uses the outdated protocol and the dependency of that software to the outdated library that provides it using the ``system migration'' view. Such multi-level migration graphs are further discussed with a concrete example in Section~\ref{sec:real-world-examples:multi-level}.
\end{description}
Note that when we consider classical, naturally grown enterprise networks, the graphs will probably be both large and complex. Also, they will have many interdependent systems. When we look at a snap-shot of a migration graph of such an infrastructure, the design choices for modeling real-world IT infrastructure graphs are vast and we discovered a variety of possible ways to model such IT infrastructures.
A reasonable assumption might be that the best strategy is to model only what is vital for the given migration and leave out all other information. Examining the exact choices, however, deserves its own study and will be considered in future work.

Independent of the chosen model, the resulting graph enables us to assess quantitative measures easily, such as fan-in and fan-out behavior of certain nodes, the actual complexity of the single migration steps in terms of the sizes of the corresponding migration cluster, the distribution of those sizes, longest paths through the graph, and many more. Also any graph optimization technique available can be potentially applied to the graph, opening the migration problem to a plethora of algorithmic analysis possibilities.

For instance, one can explore minimum vertex covers or feedback arc set formulations to identify minimal sets of components whose migration breaks critical dependency cycles. Alternatively, maximum flow algorithms may be used to model and optimize parallelizable migration paths under budget or risk constraints. Moreover, subgraph isomorphism could help detect known migration patterns or similarities within different parts of large migration graphs or help to identify frequent patterns in migration projects.

\subsection{Cryptographic Agility}
\label{sec:real-world-examples:crypto-agility}

In practice, well constructed applications which follow software design principles and adhere to best practices to minimize technical debt help to make migrations easier. Key principles include loose coupling (e.\,g., between classes, packages, modules, libraries), broad but shallow microservices to minimize latency, or clear, human-readable, and enforceable regulations. These principles lead to a robust infrastructure with fewer direct or deep dependencies, which in turn mitigates the problems of the migration process.
This is particularly relevant for network protocols, where different versions can be implemented to support both new and old cryptographic algorithms. Such designs clearly reduce the depth of dependencies in an application, as they allow for parallel migrations of different components, rather than requiring a sequential migration of many components.
Mutual cryptographic dependencies between communicating parties can be classified into three categories:

\begin{enumerate}
\item \textit{Negotiable cryptography}.
Both endpoints implement a mechanism to support algorithm negotiation (e.\,g., via a suitable cipher-suite selection as in TLS). New algorithms can be introduced without discontinuing legacy ones, so upgrades remain fully backward-compatible.

\item \textit{Unilaterally fixed / extendable cryptography}.
Client implementations employ a single, fixed algorithm, whereas servers advertise a fixed \textit{set} of algorithms that can be augmented over time. Compatibility is preserved as long as the server continues to offer every algorithm required by current clients.

\item \textit{Bilaterally fixed cryptography}.
Each component is hard-wired to a specific algorithm.
Any change therefore requires a lock-step migration of all participating components to avoid interoperability loss.
\end{enumerate}

Although dependency management is context-sensitive, this taxonomy generalizes beyond network communication. In any domain, a dependency that can be evolved while maintaining backward compatibility allows to break up migration clusters: individual components may adopt the new version independently, reducing coordination overhead and operational risk.

The need to update security procedures offers a strategic opportunity to strengthen overall system security. Administrators as well as developers and system architects should follow established best practices and standards that embed three complementary properties into the architecture~\cite{naether24sok}:

\begin{itemize}
  \item \textit{Cryptographic agility}. The ability to replace or deprecate algorithms and key sizes rapidly in response to emerging threats or standard revisions.
  \item \textit{Cryptographic versatility}. The capacity to introduce new primitives or deployment models without large-scale structural rewrites.
  \item \textit{Cryptographic interoperability}. The assurance that heterogeneous components continue to interoperate correctly after any cryptographic change.
\end{itemize}


For system administrators, classic software engineering maxims\,--\,fewer dependencies, loose coupling, and well-defined, crypto-negotiable interfaces\,--\,ease the way to cryptographic agility. When these properties are engineered together, the system can evolve seamlessly while maintaining a consistently high security maturity level. We explore these typical patterns in the following.

\subsection{Local Dependency Patterns}
\label{sec:real-world-patterns}

\begin{figure*}[t]
  \centering
  \begin{subfigure}[b]{0.19\textwidth}
    \centering
    \scalebox{.9}{%
      \begin{tikzpicture}
        \node[minimum height=5cm, inner sep=0pt] (bigNode) {
          \begin{tikzpicture}[auto,
              > = Stealth,
              every edge quotes/.style = {font=\footnotesize}, 
              every edge/.append style = {->, draw=bluelight, thick},
              every loop/.append style = {<-, looseness = 12},
              node distance = 15mm,
              state/.style = {circle, semithick, fill=bluelight, draw=bluelight, text=white, minimum size=2.5em},
              initial text = ,]

            \node (A) [state] {$v$};
          \end{tikzpicture}};
      \end{tikzpicture}}
    \caption{ }\label{fig:case1}
  \end{subfigure}
  \begin{subfigure}[b]{0.19\textwidth}
    \centering
    \scalebox{.9}{%
      \begin{tikzpicture}
        \node[minimum height=5cm, inner sep=0pt] (bigNode) {
          \begin{tikzpicture}[auto,
              > = Stealth,
              every edge quotes/.style = {font=\footnotesize}, 
              every edge/.append style = {->, draw=bluelight, thick},
              every loop/.append style = {<-, looseness = 12},
              node distance = 15mm,
              state/.style = {circle, semithick, fill=bluelight, draw=bluelight, text=white, minimum size=2.5em},
              initial text = ,]

            \node (L) [state] {$v_{1}$};
            \node (K) [state, below=.9cm of L] {$v_{2}$};
            \node (J) [state, below=.9cm of K] {$v_{3}$};
            \node (I) [state, right=of K] {$w$};

            \path   (J) edge (I)
            (K) edge (I)
            (L) edge (I)
            ;
          \end{tikzpicture}};
      \end{tikzpicture}}
    \caption{ }\label{fig:case2}
  \end{subfigure}
  \begin{subfigure}[b]{0.2\textwidth}
    \centering

    \scalebox{.9}{%
      \begin{tikzpicture}
        \node[minimum height=5cm, inner sep=0pt] (bigNode) {
          \begin{tikzpicture}[auto,
              > = Stealth,
              every edge quotes/.style = {font=\footnotesize}, 
              every edge/.append style = {->, draw=bluelight, thick},
              every loop/.append style = {<-, looseness = 12},
              node distance = 15mm,
              state/.style = {circle, semithick, fill=bluelight, draw=bluelight, text=white, minimum size=2.5em},
              initial text = ,]

            \node (F) [state] {$v$};
            \node (C) [state, right=of F] {$w_2$};
            \node (B) [state, above=.7cm of C] {$w_1$};
            \node (D) [state, below=.7cm of C] {$w_3$};

            \path   (F) edge (B)
            (F) edge (C)
            (F) edge (D)
            ;
          \end{tikzpicture}};
      \end{tikzpicture}}
    \caption{ }\label{fig:case3}
  \end{subfigure}
  \begin{subfigure}[b]{0.19\textwidth}
    \centering
    \scalebox{.9}{%
      \begin{tikzpicture}
        \node[minimum height=5cm, inner sep=0pt] (bigNode) {
          \begin{tikzpicture}[auto,
              > = Stealth,
              every edge quotes/.style = {font=\footnotesize}, 
              every edge/.append style = {->, draw=bluelight, thick},
              every loop/.append style = {<-, looseness = 12},
              node distance = 15mm,
              state/.style = {circle, semithick, fill=bluelight, draw=bluelight, text=white, minimum size=2.5em},
              initial text = ,]

            \node (E) [state] {$v$};
            \node (F) [state, below=of E] {$w$};

            \path   (F) edge [bend left=10] (E)
            (E) edge [bend left=10] (F)
            ;
          \end{tikzpicture}};
      \end{tikzpicture}}
    \caption{ }\label{fig:case4}
  \end{subfigure}
  \begin{subfigure}[b]{0.19\textwidth}
    \centering
    \scalebox{.9}{%
      \begin{tikzpicture}
        \node[minimum height=5cm, inner sep=0pt] (bigNode) {
          \begin{tikzpicture}[auto,
              > = Stealth,
              every edge quotes/.style = {font=\footnotesize}, 
              every edge/.append style = {->, draw=bluelight, thick},
              every loop/.append style = {<-, looseness = 12},
              node distance = 15mm,
              state/.style = {circle, semithick, fill=bluelight, draw=bluelight, text=white, minimum size=2.5em},
              initial text = ,]

            \node (D) [state] {$v_1$};
            \node (B) [state, above right=.9cm of D] {$v_2$};
            \node (C) [state, below right=.9cm of D] {$v_3$};
            \path   (D) edge[bend left=10] (B)
            (B) edge[bend left=10] (C)
            (C) edge[bend left=10] (D)
            ;
          \end{tikzpicture}};
      \end{tikzpicture}}
    \caption{ }\label{fig:case5}
  \end{subfigure}


  \caption{Dependency sub-graphs of real-world migration attempts: \subref{fig:case1} Single component with no dependencies. \subref{fig:case2} Many components depend on a single component. \subref{fig:case3} One component depends on many components. \subref{fig:case4} Mutual dependencies between two components. \subref{fig:case5} Higher-order cycles.}\label{fig:real-world-patterns}
\end{figure*}

When analyzing the occurrence of graph patterns in the real world, we begin
with the kinds of \emph{typical} dependency patterns we observe in practice. To do
so, we depicted different base-patterns from which the graph can be constructed
in Fig.~\ref{fig:real-world-patterns}. In the following, we will express real-world
observations related to these patterns and discuss their semantics and meaning for our
introduced theory. Also, we briefly reflect on their impact on post-quantum migration.

\subsubsection{API-compatible change of algorithms (Fig.~\ref{fig:case1})}
When migrating an isolated (ge\-ne\-ric) technological
component $v$, a simple change in the code or even an update of the relevant software
will successfully migrate the component. This corresponds to the case in
Fig.~\ref{fig:case1}. In the context of post-quantum migration this could be a simple
of the key-length of symmetric ciphers such as AES~\cite{nist-sp-800-57-pt1-r5}.

\subsubsection{Operating System (OS) dependencies (Fig.~\ref{fig:case2})}
Many security-critical applications ($v_1$, $v_2$, $v_3$, $\dots$) depend on a single OS ($w$) to provide core functionality such as cryptographic services, driver support, or security enforcement. Migration is constrained when the OS lacks support for updated cryptographic mechanisms. A key example is the transition to post-quantum secure boot: before applications relying on verified boot integrity can migrate, the OS must first support PQC-based firmware and kernel signature verification. This highlights how OS constraints can delay cryptographic migration across multiple dependent components.

\subsubsection{Regulatory requirements (Fig.~\ref{fig:case3})}
When regulatory requirements $w_1$, $w_2$, $w_3$, $\dots$ impose restrictions on the migratability of certain technical components $v$, only a change in the regulatory requirements enables us to actually migrate. Changing regulatory requirements, however, often depends on many further, often non-technical, components, such as laws, standardization, or technological standards in general. This directly applies to post-quantum migration as well, when we want conformance to one of the PQC NIST standards such as FIPS 203~\cite{fips203} for example.

\subsubsection{Communication protocols (Fig.~\ref{fig:case3}/\ref{fig:case4})}
When migrating communication protocols from one version to another, it is vital that all communicating parties in the migration cluster are migrated at the same time. This is in particular a problem when there are many parties involved and only one part of the IT infrastructure may be updated conveniently. It should be noted that such dependencies typically do not form transitive clusters, since cipher suites can be agreed upon by two respective communicating parties only. In the context of post-quantum migration this could be the replacement of TLS by a post-quantum-ready version. 

\begin{figure*}[t]
  \hspace*{-0.5cm}
  \resizebox{\textwidth}{!}{\usetikzlibrary{positioning,arrows.meta,shapes.symbols,fit,calc}

\tikzset{
  every node/.style={font=\ttfamily},
  svcbox/.style={
    draw, rounded corners, thick,
    minimum width=3.2cm, minimum height=1cm,
    align=center
  },
  svccloud/.style={
    draw, shape=cloud, thick,
    cloud puffs=12, cloud ignores aspect,
    minimum width=3cm, minimum height=1.8cm,
    align=center,
    font=\rmfamily
  },
  arr/.style={-{Latex[length=3mm,width=2mm]}, thick},
  proto/.style={font=\small, midway, fill=white, inner sep=1pt, align=center}
}

\begin{tikzpicture}

\node[svccloud] (web) {Web\\Clients};
\node[svccloud, right=5cm of web] (ssh) {SSH\\Clients};

\node[svcbox, below=2cm of web]          (nginx)   {Nginx\\(cluster)};
\node[svcbox, below=2.5cm of nginx]        (workhorse) {Workhorse\\(cluster)};
\node[svcbox, left =8cm of nginx]       (pages)   {Pages\\(cluster)};
\node[svcbox, right=5cm of workhorse]      (shell)   {Shell\\(cluster)};

\node[svcbox, below=2.5cm of workhorse, xshift=-2.8cm] (puma)   {Puma\\(Rails cluster)};
\node[svcbox, below=2.5cm of workhorse, xshift= 2.8cm] (sidekiq){Sidekiq\\(jobs cluster)};

\node[svcbox, below=2.5cm of puma,   xshift=-2.8cm] (gitaly){Gitaly\\(cluster + praefect)};
\node[svcbox, below=4cm of sidekiq,xshift=4.5cm, yshift=-2cm] (redis) {Redis Cluster\\ (cache / sessions)};


\node[svcbox, left=2cm of redis] (postgres){PostgreSQL\\(managed)};
\node[svcbox, below=2.5cm of redis] (registry){Container\\Registry};
\node[svcbox, right=3cm of registry] (runner)  {GitLab Runner\\(autoscale)};

\node[svcbox, right=2cm of shell] (prom){Prometheus};
\node[svcbox, below=2cm of prom]  (alert){Alertmanager};
\node[svcbox, above=2cm of prom] (grafana){Grafana};

\node[svcbox, below=2.5cm of pages] (ldap){LDAP / SAML};
\node[svcbox, below=2cm of ldap] (smtp){SMTP Server};
\node[svcbox, left=2.5cm of registry] (objectstore){Storage\\(S3 / MinIO)};

\draw[arr] (web) -- node[proto, yshift=0.2cm]{HTTP\\TLS} (nginx);
\draw[arr] (ssh) -- node[proto]{SSH}          (shell);

\draw[arr] (nginx) -- node[proto] {HTTP\\TLS} (workhorse);
\draw[arr] (nginx) to[out=200,in=20] node[proto]{HTTP\\TLS} (pages);
\draw[arr] (shell) to[out=200,in=20] node[proto, xshift=-0.5cm]{API\\HTTP\\TLS} (workhorse);

\draw[arr] (workhorse) to[out=195,in=90] node[proto]{HTTP\\TLS} (puma.north);
\draw[arr] (workhorse) to[out=345,in=90] node[proto]{HTTP\\TLS} (sidekiq.north);

\draw[arr] (puma)    -- node[proto]{gRPC\\8075/TLS}      (gitaly);
\draw[arr] (sidekiq) -- node[proto, xshift=-1.8cm, yshift=-0.7cm]{gRPC\\8075/TLS}      (gitaly);

\draw[arr] (puma.east)    to[out=10,in=190]  node[proto]{Redis\\6379/TLS} (redis.west);
\draw[arr] (sidekiq) -- node[proto]{Redis\\6379/TLS}                 (redis);

\draw[arr] (puma.south)    -- node[proto]{PostgreSQL\\5432/TLS} (postgres.north);
\draw[arr] (sidekiq) to[out=-90,in=20] node[proto, yshift=-.2cm]{PostgreSQL\\5432/TLS} (postgres);

\draw[arr] (gitaly.east)  to[out=0,in=180] node[proto]{S3\\HTTP\\TLS} (objectstore.west);

\draw[arr] (registry.west) to node[proto]{S3\\HTTP\\TLS} (objectstore.east);

\draw[arr] (runner.west)  to node[proto]{HTTP\\TLS}    (registry.east);
\draw[arr] (runner.north) to[out=140,in=-50] node[proto]{SSH} (shell.south);
\draw[arr] (runner.north) to[out=140,in=-30] node[proto, xshift=1cm, yshift=-2.7cm]{HTTPS\\TLS} (nginx.east);

\draw[arr] (prom)          -- node[proto]{Alerts\\HTTP\\TLS} (alert);
\draw[arr] (grafana.south) -- node[proto]{Queries\\HTTP\\TLS} (prom);

\draw[arr] (prom.north) to[out=160,in=0] node[proto, yshift=0.5cm, xshift=-2cm]{/metrics\\HTTP\\TLS} ($(nginx.east)+(0,0.2cm)$);

\draw[arr] (puma.west) to[out=180,in=0]    node[proto]{LDAP\\636/TLS}  (ldap.east);
\draw[arr] (puma.west) to[out=200,in=30]   node[proto]{SMTP\\465/TLS}  (smtp.east);


\node[svcbox, below=2.2cm of alert]  (vault){Vault};
\node[svcbox, below=1.9cm of vault]  (cert){Cert-Manager};

\draw[arr] (puma.north) to[out=45,in=160] node[proto, xshift=2.5cm, yshift=-1.2cm]{HTTP\\TLS} (vault.west);
\draw[arr] (cert.north) to node[proto]{HTTP\\TLS} (vault.south);

\node[draw, dashed, draw=bluelight, text=white, rounded corners, fit=(nginx) (pages), inner sep=7mm, label=below:\Large$C_{~nginx}$] {};

\node[draw, dashed, draw=bluelight, text=white, rounded corners, fit=(workhorse) (puma) (sidekiq) (gitaly), inner sep=7mm, label={[xshift=-1cm]below:\Large$C_{~core}$}] {};

\node[draw, dashed, draw=bluelight, text=white, rounded corners, fit=(prom) (alert) (grafana), inner sep=7mm, label=below:\Large$C_{~mon}$] {};

\node[draw, dashed, draw=bluelight, text=white, rounded corners, fit=(vault) (cert), inner sep=7mm, label=below:\Large$C_{~id}$] {};

\node[draw, dashed, draw=bluelight, text=white, rounded corners, fit=(ldap) (smtp), inner sep=7mm, label=below:\Large$C_{~aux}$] {};

\node[draw, dashed, draw=bluelight, text=white, rounded corners, fit=(redis) (postgres), inner sep=7mm, label=below:\Large$C_{~ds}$] {};

\node[draw, dashed, draw=bluelight, text=white, rounded corners, fit=(objectstore) (registry) (runner), inner sep=7mm, label=below:\Large$C_{~cicd}$] {};

\node[draw, dashed, draw=bluelight, text=white, rounded corners, fit=(shell), inner sep=7mm, label=below:\Large$C_{~ssh}$] {};

\node[minimum width=15cm, below=of smtp, yshift=-4.8cm, xshift=3cm] (legende) {%
\begin{tabularx}{12cm}{m{1.5cm} m{4cm} X}
\toprule
\textrm{Cluster} & \textbf{Name (Size)} & \texttt{Services} \\
\arrayrulecolor{black}\midrule\arrayrulecolor{black}
$C_{~nginx}$ & \textbf{Nginx Ingress (2)} & \texttt{Nginx}, \texttt{Pages} \\
$C_{~core}$  & \textbf{Core App (4)} & \texttt{Gitaly}, \texttt{Puma}, \texttt{Sidekiq}, \texttt{Workhorse}\\
$C_{~mon}$   & \textbf{Monitoring (3)} & \texttt{Alertmanager}, \texttt{Grafana}, \texttt{Prometheus} \\
$C_{~id}$    & \textbf{Identity (2)} & \texttt{Cert-Manager}, \texttt{Vault}\\
$C_{~ds}$    & \textbf{Data Stores (2)} & \texttt{PostgreSQL}, \texttt{Redis Cluster} \\
$C_{~aux}$   & \textbf{Aux. Services (2)} & \texttt{LDAP}, \texttt{SMTP} \\
$C_{~cicd}$  & \textbf{CI/CD (3)} & \texttt{Container Registry}, \texttt{GitLab Runner}, \texttt{Storage} \\
$C_{~ssh}$   & \textbf{SSH Services (1)} & \texttt{Shell} \\
\bottomrule
\end{tabularx}};

\end{tikzpicture}}%

  \caption{A simplified representation of a migration graph for \gitlab; for the sake of visual clarity, several ancillary communication edges, such as specific interactions with downstream services, have been intentionally omitted.\label{fig:example_real-world_graph}}
\end{figure*}

\subsubsection{Dependency on technical requirements (Figs.~\ref{fig:case4}/\ref{fig:case5})}
Occasionally, there are requirements on the amount of memory, the capabilities of computing platforms, or network throughput restrictions. In the context of post-quantum migration this applies when, e.\,g., the key sizes of post-quantum algorithms are much larger than those of classical algorithms. To successfully migrate, the additional computational resources must be provided.

\subsubsection{Dependencies of software package distribution for, e.\,g., operating systems, update management or software development (Figs.~\ref{fig:case1}/\ref{fig:case2}/\ref{fig:case3})}
When considering software packages and their dependencies, the graph formed is a directed acyclic graph with no mutual dependencies. This graph represents combinations of cases illustrated in Figs.~\ref{fig:case1},~\ref{fig:case2}, and~\ref{fig:case3}. The migration of such a case can be done corresponding to the canonical topological order of this graph. In the context of post-quantum migration this could be the replacement of cryptographic libraries with post-quantum-ready versions.

\subsubsection{Cyclically depending assets (Fig.~\ref{fig:case5})}
If, e.\,g., three \emph{cryptographic assets} are in a dependency relationship with each other, functionality issues may occur when a single asset is migrated~\cite{etsi24cyberqsc}. The functionality issues might cascade and affect all assets in the circle, even the migrated one. The problem is mitigated in our model as we assume that such circular clusters of assets can only be migrated when all assets are migrated at the same time. In the context of post-quantum migration this might happen, e.\,g., when we observe circularly cross-signed certificates.
Circular cases are considered rare in practice, because they create dependencies that are avoided in software architecture and engineering. However, studies show that they still occur in several cases e.\,g., when looking at class dependencies~\cite{Melton2007classcycles} or microservice dependencies\cite{Genfer2021cyclicmicroservices}. Developers can address those kinds of dependencies e.\,g. via Anti-Corruption Layers~\cite{Evans2004dddACL}. 


\subsubsection{Further patterns}
Although these patterns are complete from a graph-theoretic point of view, e.\,g., any directed graph can be (non-uniquely) decomposed into sub-graphs with exactly the patterns from Fig. ~\ref{fig:real-world-patterns}, it is important to note that each of them must be handled differently in each context. For example, a circular dependency for software packages would be resolved by updating all software components at once. A circular regulatory requirement update might be resolved by changing the law, then a standard, then a communication protocol and impose further regulatory requirement updates.
We encourage the inclined reader to find more naturally occurring real-world examples. Using these base-patterns, a wide range of possible security scenarios can not only be modeled, but also described and simplified.


\section{Real-World Examples}
\label{sec:real-world-examples}

Migration graphs can help in understanding the complexity of the system involved and in considering the time needed to migrate, but also quantitatively specify the number of components and the various dependencies. As we have seen, these dependencies can exist between single components, between interdependent clusters of components, and between external factors such as regulatory requirements, legal issues, or third-party components. The complexity of migrations requires expert knowledge and sufficient dedicated resources, resulting in significant costs.
In the following, we explore two practical examples\,--\,one from software development and one for critical infrastructures\,--\,and explain how a potential migration graph model could look like and how more complex systems might require a multi-layered approach.

Given the general structure of practical migration, we believe our model can significantly support and enhance existing migration strategies and tools, particularly in the standardization of migration projects through guidelines, rules, and best practices. These models can optimize the process, provide practical examples for guidance, and reduce cost and effort. This potentiall can help system administrators to make large migration projects easier to plan and conduct.

\subsection{Migration Graphs for Open-Source Software Projects}
\label{sec:real-world-examples:gitlab}


We investigate the migration of an open-source project to PQC using the CI/CD DevOps pipeline of the web-based DevOps platfrom \gitlab\footnote{\url{https://perma.cc/X7DC-SYYL}} as a representative case. The \gitlab platform comprises numerous subsystems that exchange data over cryptographically protected channels\,--\,primarily HTTPS/TLS and SSH\,--\,as well as auxiliary services such as LDAP, SMTP, and RPCs (cf.\ Fig.~\ref{fig:example_real-world_graph}). The figure reflects the architecture of a large-scale production deployment of \gitlab serving a substantial user base: vertices correspond to individual components, while directed edges denote protocol-level dependencies. By construction, the resulting graph is a directed acyclic graph. However some of the connections can be considered two-way, when for example cryptographic key-exchanges are performed. Grouping these components semantically yields canonical migration clusters.
In this scenario, the components no longer communicate over Unix Sockets (the default for \gitlab), but rather TCP-Sockets, allowing \gitlab to run as a distributed system on different machines. For communication different cryptographic protocols. Although the graph has been slightly simplified for sake of exposition, it still contains most of the relevant components and their dependencies.

If we inspect Fig.~\ref{fig:example_real-world_graph}, we can see that there are several chains of dependent services, e.\,g., \texttt{WebClients} $\Rightarrow$ \texttt{Nginx} $\Rightarrow$ \texttt{Workhorse} $\Rightarrow$ \texttt{Puma} $\Rightarrow$ \texttt{PostgreSQL}. When leaving out the external WebClients, we see that the longest path through the graph is 4 hops long, i.\,e., has length 3. This means that, when strictly following the rules of the migration model proposed, we would have to migrate at least 4 components in a row. The same applies to the other paths through the graph, e.\,g., SSH Clients $\Rightarrow$ \texttt{Shell} $\Rightarrow$ \texttt{Gitaly} $\Rightarrow$ \texttt{Object Storage}.
However, in several practical cases, such as networking, depth dependencies may be divided into migration clusters, only considering direct neighbors. This is because we can upgrade TLS nodes to a higher version (or a different cipher-suite) for directly communicating components, independent of the other components. Eventually, all services will speak the new version of TLS and the migration has been fully carried out. This fits exactly to our formal model and illustrates that in this scenario, some depth dependencies can be converted into migration clusters for successive migration of the protocols involved.

We now present a detailed analysis of the migration clusters within the \gitlab architecture as defined in the lower left table in Fig.\ref{fig:example_real-world_graph}. Each cluster groups together services that share common dependencies and are likely to be migrated as a unit. Units are grouped by their functionality and the dependencies between them. The clusters are named after the main purpose they serve, and the number of services in each cluster is indicated in parentheses. As mentioned in Section~\ref{sec:real-world-examples:gitlab}, in specific cases, the migration of dependencies may be split into parallel dependencies, which can be migrated independently. 
The graph of Fig.~\ref{fig:example_real-world_graph} induces the dependencies between those clusters and the corresponding migration graph is given in Fig.~\ref{fig:tikz-migration-graph}.
\begin{figure}[t]
  \centering
  \resizebox{0.98\columnwidth}{!}{\usetikzlibrary{positioning,arrows.meta}

\tikzset{
 svcbox/.style={
    draw, rounded corners,
    minimum width=1cm, minimum height=1cm,
    align=center
  },
  arr/.style={-{Latex[length=3mm,width=2mm]}, thick},
}

\begin{tikzpicture}

\node[svcbox] (core) {$C_{~core}$};

\node[svcbox, left=of core] (ds) {$C_{~ds}$};
\node[svcbox, above= of ds, yshift=-0.5cm] (id) {$C_{~id}$};
\node[svcbox, below= of ds, yshift=0.5cm] (aux) {$C_{~aux}$};

\node[svcbox, right=of core, yshift=1cm] (ssh) {$C_{~ssh}$};
\node[svcbox, below= of ssh] (nginx) {$C_{~nginx}$};
\node[svcbox, right=of nginx] (mon) {$C_{~mon}$};
\node[svcbox, right=of ssh] (cicd) {$C_{~cicd}$};

\draw[arr] (core) -- (ds);
\draw[arr] (core) -- (id);
\draw[arr] (core) -- (aux);

\draw[arr] (nginx) -- (core);
\draw[arr] (ssh) -- (core);

\draw[arr] (mon) -- (nginx);
\draw[arr] (cicd) -- (ssh);
\draw[arr] (cicd) -- (nginx);

\end{tikzpicture}}%

  \caption{Condensed migration graph of the \gitlab system from Fig.~\ref{fig:example_real-world_graph}.}
  \label{fig:tikz-migration-graph}
\end{figure}

%
%
We can now measure empirically the predicted properties from Table~\ref{tab:numerical_behavior} for this graph, such as the number of clusters, the mean of the cluster sizes and its standard deviation, as well as the depth of the condensation (cf.\ Table~\ref{tab:migration_clusters_metrics}). We note that the result fits well to the theoretical predictions.

\begin{table}[h!]
\caption{Comparison of theoretical predictions from Table~\ref{tab:numerical_behavior} and the observed migration graph metrics for $n=19$.}
\label{tab:migration_clusters_metrics}
\centering
\begin{tabularx}{\columnwidth}{X*{3}{>{\centering\arraybackslash}m{1.2cm}}}
\toprule
\makecell{\textbf{Metric}} &
\makecell{\textbf{Formula}\\\textbf{(Table~\ref{tab:numerical_behavior})}} &
\makecell{\textbf{Expected}\\\textbf{($n=19$)}} &
\makecell{\textbf{Observed}\\\phantom{\textbf{($n=19$)}}}\\
\arrayrulecolor{black}\midrule\arrayrulecolor{black}
Cluster size $\mathsf{E}(\#c(v))$ & $\log n$ & 2.94 & $2.38$ \\
\arrayrulecolor{black!20}\midrule\arrayrulecolor{black}
Cluster count $\mathsf{E}(\#(G/c))$ & $\dfrac{n}{\log n}$ & 6.45 & 8 \\
\arrayrulecolor{black!20}\midrule\arrayrulecolor{black}
Std.\ dev.\ $\mathsf{SD}(\#(G/c))$ & $\dfrac{\sqrt{n}}{\log n}$ & 1.48 & $0.92$ \\
\arrayrulecolor{black!20}\midrule\arrayrulecolor{black}
Migration depth (optim.) & $\sqrt{n}$ & 4.36 & 3 \\
\arrayrulecolor{black!20}\midrule\arrayrulecolor{black}
Migration depth (cons.) & $\dfrac{n}{\log^{2} n}$ & 2.19 & 3 \\
\bottomrule
\end{tabularx}
\end{table}

In both, the original and the clustered graph, we observe patterns like the ones described in Fig.~\ref{fig:real-world-patterns}. Focusing on the condensed graph in Fig.~\ref{fig:tikz-migration-graph}, we can see that $C_{~core}$ corresponds to the fan-in pattern (Fig.~\ref{fig:case2}), with cluster $C_{~ssh}$ and $C_{~nginx}$, but also to the fan-out pattern (Fig.~\ref{fig:case3}), with cluster $C_{~id}$, $C_{~ds}$, and $C_{~aux}$. Depending on how you actually perform the migration, a connection from one cluster to another or even one service to another may be seen as mutual dependencies as in Fig.~\ref{fig:case4}. A single node dependency like in Fig.~\ref{fig:case1} may be found in the $C_{~ds}$ cluster, where \texttt{PostgreSQL} uses encryption for its data at rest, which is a single point of migration. A circular dependency can not be found in the graph, as this is avoided in the architecture of \gitlab. However, it is important to note that circular dependencies can occur in other systems, and our model allows for them as well.
%
%
\gitlab already has an architecture which is well designed in terms of cryptographic agility, as explained in Section~\ref{sec:real-world-examples:crypto-agility}. Nevertheless, the graph model not only analyses how to plan a PQC migration but also offers a concrete optimization target: design for minimal depth and small, acyclic clusters so that future cryptographic transitions become routine maintenance tasks rather than multi-year projects.

\subsection{Multi-level Migration Graphs for Critical Infrastructure}
\label{sec:real-world-examples:multi-level}

Critical infrastructures are an especially important class of systems since their security has a direct impact on everyday life. They include healthcare, telecommunications, water supply, and many other examples. These systems can be particularly vulnerable and a number of attacks on them have been reported~\cite{alcaraz2015critical}. At the same time, they can be very complex and are often distributed. We believe that migration graphs are therefore particularly useful tools for critical infrastructure when considering the evolution of the cryptographic mechanisms used in them.

As an example of critical infrastructure, we are going to discuss modern power distribution grids. Power grids are a critical infrastructure in any modern society and have been the target of cyber attacks, e.\,g., in the war between Russia and Ukraine\cite{Kostyuk2019cyberattacks}. The increased digitalization and the move towards more sophisticated monitoring infrastructure deployed in the grid in a distributed fashion opens many new attack paths~\cite{tatipatri2024comprehensive}. Tatipatri and Arun~\cite{tatipatri2024comprehensive} provide a taxonomy of possible attacks, ranging from injecting false data into measurement system to physical attacks on the equipment. Generally, the aim of such attacks is to disrupt the service of the power system and stop the flow of power to consumers, often by destroying the power equipment itself.

\begin{figure*}
  \includegraphics[width=\textwidth]{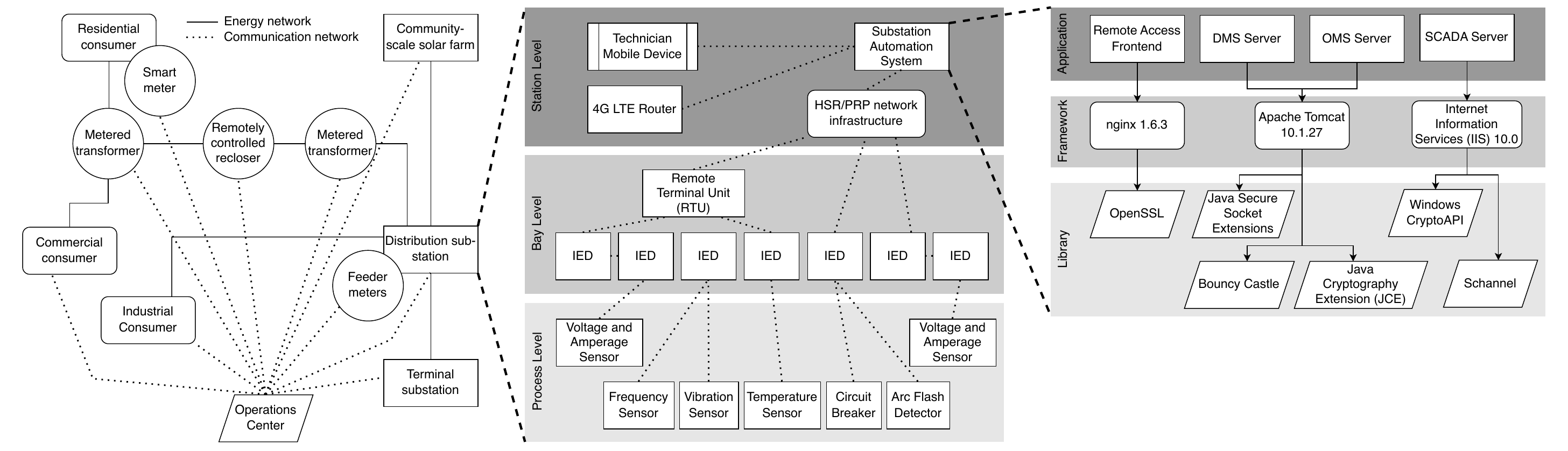}
  \caption{The multi-layered architecture in a distribution power grid. Each of the levels of the grid can be seen as a distinct migration context with its own migration graph. The terminology used in the figure is explained in more detail in the text.}
  \label{fig:power-grid-levels}
\end{figure*}

Fig.~\ref{fig:power-grid-levels} shows three levels of the architecture of the power distribution grid. On the left-hand side is the power distribution grid that consists of many physical devices that are located along the power lines. They include smart meters in homes, small pole-mounted transformers and reclosers (i.e., breakers that can be remotely controlled), more complex monitoring equipment for larger consumers and, crucially, distribution substations. All of these components are connected to an operation center, run by the grid provider, via different means, including public 4G networks, fiber optics, or specialized wireless technology such as WiMAX~\cite{ieee802.16e-2005}. This means that a migration has to consider this level and it is possible to draw a migration graph here.

On the substation level, we encounter a different, complex architecture. The substation automation system is connected to intelligent electronic devices (IEDs) via a high-redundancy network. The IEDs in turn are connected to the sensors and actuators that operate directly on the electrical equipment. Some IEDs are connected to remote terminal units (RTU) that combine their information and make it accessible to the substation automation system. At the same time, the distribution station can be serviced by a technician with a mobile device such as a laptop that can connect to the substation automation system via a wired connection or wirelessly. In our example, the substation can be reached remotely via a 4G LTE router. Substations in the power grid use specialized protocols described in IEC~61968~\cite{iec61968} and security mechanisms described in IEC~62351~\cite{iec62351}. Here, again, a migration graph can be drawn up to describe the migration that is necessary to migrate an entire substation to more modern cryptographic mechanisms as shown in the middle part of Fig.~\ref{fig:power-grid-levels}.

Finally, we can look at one concrete subsystem within the substation (right hand side of Fig.~\ref{fig:power-grid-levels}). As an example, the substation automation system consists of several servers that run on dedicated hardware. The Remote Access Frontend provides access to technicians or remote workers via a web interface. It uses \texttt{nginx} as a webserver which in turn uses OpenSSL to provide cryptographic mechanisms. The Distribution Management Server (DMS) deals with optimizing the electricity flow while the Outage Management System (OMS) kicks in if some equipment fails or the power flow is interrupted. They both run on an Apache Tomcat application server and use libraries that come with the Java Runtime Environment as well as the open source library Bouncy Castle to provide cryptographic functions. Finally, the Supervisory Control and Data Acquisition (SCADA) system which controls the overall processes and the data collection is based on Internet Information Services (IIS) and the crypto libraries embedded in the Microsoft Windows operating system. Again, migrating the substation automation system is a major task itself which can be described, analyzed, and planned with migration graphs.

Detailing the migration graphs for this complex example will be the topic of a forthcoming paper. What the power distribution system as a representative of a critical infrastructure shows, however, is the need to treat migration graphs as multi-layered graphs where each layer can be considered as one specific migration context. In such complex systems, a node in a migration graph can be decomposed to reveal an underlying migration problem in a different context that can be analyzed with the same tools (cf.~Section~\ref{sec:migration-modelling}). Decomposition and abstraction in this sense are powerful mechanisms to reduce the complexity of migration graphs. They allow splitting up the work; e.\,g., the substation automation system can be migrated by a contractor. But abstraction also allows simplifying the graphs for different purposes such as effort estimation or visualization. We will explore the power of this hierarchical multi-graph paradigm in future work.

In summary, managing the intrinsic complexity of large migration projects involves understanding the dependencies and strategically planning the different migration steps, while keeping the economic implications in mind. Supporting companies with standardized guidelines and open-source tools is crucial, as is fostering an agile approach to sustainable security solutions. By addressing these key areas, organizations can overcome migration challenges and maintain robust security within their IT infrastructures.

\section{Conclusion}
\label{sec:conclusion}

In this work, we use classical mathematical results from graph theory, probability theory, and combinatorial analysis to assess the difficulty of a very practical problem: cryptographic migration of large IT infrastructures. We formulated a formal model of cryptographic migrateability in terms of directed graphs and showed several expected properties of these graphs. This approach not only elevates the discussion to a scientific level but also establishes a foundation for further research on more efficient migration strategies. 

We showed that migration projects always have an intrinsic complexity due to many dependent, comparatively small migration clusters which successively depend on each other.

We also discussed the consequences of this complexity, by pointing out recurring migration patterns as well as single-layered and multi-layered, practical examples. We distinguish different migration contexts and discuss our work in the context of existing frameworks and standards for migrating cryptographic assets as well as in the context of cryptographic agility.

Our work sets the stage for future research aimed at refining these models, integrating them with real-world data, and ultimately creating more adaptive, user-friendly tools that support real-world migration. We therefore aim here to scientifically substantiate migration approaches. By emphasizing the relevance of this formal approach, we hope to support the development of more robust migration strategies and (open-source) tooling that can keep pace with the evolving landscape of cryptographic standards.

\subsection*{Future Work}
Still, there is a lot of work to be done when it comes to practical migration projects.
\begin{description}
  \item[Obtaining useful migration graphs in practice:] Modelling real-world IT infrastructures or relevant subsets will clearly be a challenging task. The resulting migration graphs should neither be trivial nor too vast. In addition, migration graphs change over time as the system evolves while the migration is ongoing.
  \item[Multi-layered migration graphs:] As demonstrated in our example of the power distribution grid, real-world systems have multiple layers which are defined by their migration contexts. It is worthwhile to treat these contexts separately, e.\,g., by looking at a single server as well as at the more complex network in which it is integrated. However, in practice, there are dependencies between the contexts that need to be studied in more detail.
  \item[Migration patterns:]  Applying our model to real-world migration projects will provide an empirical basis to identify specific, reoccurring cases that can be addressed with a set of migration patterns.
  \item[Migration strategies:] While we have hinted at possible strategies to migrate specific patterns in this paper, a more thorough technical analysis is required to make concrete suggestions about migration strategies that take the technical constraints (such as limited upgradeability, memory, or CPU power) into account.
  \item[Migration cost:]  As migration problems typically come with some kind of measurable cost per migration step, a refinement of our model would result in migration graphs where optimal migration strategies are not as easy to obtain as here.
  \item[Tooling:] Practical migrations are too large to handle with a purely manual approach. Our formal model serves as a starting point for tool-supported analysis of systems and tool-supported planning of migrations. We are already working on tooling that creates graph structures out of cryptographic inventories which will have the ability to calculate migrations graphs using the principles laid out in this paper.
\end{description} 

\ifanonymous
\else
\section*{Acknowledgment}

The authors thank Tobias J. Bauer for his valuable comments during the genesis of this work.
\fi


%

%
%
\bibliographystyle{IEEEtran}
\bibliography{\bibdir bibliography,\bibdir rfc}

\begin{thebibliography}{10}
\providecommand{\url}[1]{#1}
\csname url@samestyle\endcsname
\providecommand{\newblock}{\relax}
\providecommand{\bibinfo}[2]{#2}
\providecommand{\BIBentrySTDinterwordspacing}{\spaceskip=0pt\relax}
\providecommand{\BIBentryALTinterwordstretchfactor}{4}
\providecommand{\BIBentryALTinterwordspacing}{\spaceskip=\fontdimen2\font plus
\BIBentryALTinterwordstretchfactor\fontdimen3\font minus
  \fontdimen4\font\relax}
\providecommand{\BIBforeignlanguage}[2]{{%
\expandafter\ifx\csname l@#1\endcsname\relax
\typeout{** WARNING: IEEEtran.bst: No hyphenation pattern has been}%
\typeout{** loaded for the language `#1'. Using the pattern for}%
\typeout{** the default language instead.}%
\else
\language=\csname l@#1\endcsname
\fi
#2}}
\providecommand{\BIBdecl}{\relax}
\BIBdecl

\bibitem{barpol2021nist}
\BIBentryALTinterwordspacing
W.~Barker, W.~Polk, and M.~Souppaya, ``{Getting Ready for Post-Quantum
  Cryptography: Exploring Challenges Associated with Adopting and Using
  Post-Quantum Cryptographic Algorithms},'' National Institute of Standards and
  Technology, Tech. Rep., April 2021. [Online]. Available:
  \url{https://doi.org/10.6028/NIST.CSWP.04282021}
\BIBentrySTDinterwordspacing

\bibitem{Patil16}
\BIBentryALTinterwordspacing
P.~Patil, P.~Narayankar, {Narayan D.G.}, and {Meena S.M.}, ``{A Comprehensive
  Evaluation of Cryptographic Algorithms: DES, 3DES, AES, RSA and Blowfish},''
  \emph{Procedia Computer Science}, vol.~78, pp. 617--624, 2016. [Online].
  Available: \url{https://doi.org/10.1016/j.procs.2016.02.108}
\BIBentrySTDinterwordspacing

\bibitem{bsi24statusQC}
\BIBentryALTinterwordspacing
F.~Wilhelm, R.~Steinwandt, D.~Zeuch, P.~Lageyre, and S.~Kirchhoff, ``{Status of
  quantum computer development},'' {Federal Office for Information Security},
  Tech. Rep., August 2024. [Online]. Available:
  \url{https://www.bsi.bund.de/SharedDocs/Downloads/DE/BSI/Publikationen/Studien/Quantencomputer/Entwicklungstand_QC_V_2_1.html}
\BIBentrySTDinterwordspacing

\bibitem{mosca2018-will-we-be-ready}
\BIBentryALTinterwordspacing
M.~Mosca, ``{Cybersecurity in an Era with Quantum Computers: Will We Be
  Ready?}'' \emph{IEEE Security \& Privacy}, vol.~16, no.~5, pp. 38--41, 2018.
  [Online]. Available: \url{https://doi.org/10.1109/MSP.2018.3761723}
\BIBentrySTDinterwordspacing

\bibitem{pqc_adoption}
\BIBentryALTinterwordspacing
J.~Sowa, B.~Hoang, A.~Yeluru, S.~Qie, A.~Nikolich, R.~Iyer, and P.~Cao,
  ``{{Post-Quantum Cryptography (PQC) Network Instrument: Measuring PQC
  Adoption Rates and Identifying Migration Pathways}},'' 2024. [Online].
  Available: \url{https://arxiv.org/abs/2408.00054}
\BIBentrySTDinterwordspacing

\bibitem{fips203}
\BIBentryALTinterwordspacing
G.~M. Raimondo and L.~E. Locascio, ``{FIPS-203: Module-Lattice-Based
  Key-Encapsulation Mechanism Standard},'' \emph{Federal Information Processing
  Standards Publication}, vol. 203, pp. 1--56, 2025. [Online]. Available:
  \url{https://nvlpubs.nist.gov/nistpubs/FIPS/NIST.FIPS.203.pdf}
\BIBentrySTDinterwordspacing

\bibitem{fips204}
\BIBentryALTinterwordspacing
------, ``{FIPS-204: Module-Lattice-Based Digital Signature Standard},''
  \emph{Federal Information Processing Standards Publication}, vol. 204, pp.
  1--65, 2025. [Online]. Available:
  \url{https://nvlpubs.nist.gov/nistpubs/FIPS/NIST.FIPS.204.pdf}
\BIBentrySTDinterwordspacing

\bibitem{fips205}
\BIBentryALTinterwordspacing
------, ``{FIPS-205: Stateless Hash-Based Digital Signature Standard},''
  \emph{Federal Information Processing Standards Publication}, vol. 205, pp.
  1--61, 2025. [Online]. Available:
  \url{https://nvlpubs.nist.gov/nistpubs/FIPS/NIST.FIPS.205.pdf}
\BIBentrySTDinterwordspacing

\bibitem{bsi_migration}
\BIBentryALTinterwordspacing
{Federal Office for Information Security}, ``{Migration to Post Quantum
  Cryptography: Recommendations for action by the BSI},'' {Federal Office for
  Information Security}, Tech. Rep., May 2021. [Online]. Available:
  \url{https://www.bsi.bund.de/SharedDocs/Downloads/EN/BSI/Crypto/Migration_to_Post_Quantum_Cryptography.html}
\BIBentrySTDinterwordspacing

\bibitem{bsi-tr-02102-1}
------, ``{Cryptographic Mechanisms: Recommendations and Key Lengths},'' {BSI},
  Technical Guideline BSI TR-02102-1, January 2025.

\bibitem{PQC-migration-handbook-nl2023}
\BIBentryALTinterwordspacing
A.~Amadori, T.~Attema, M.~Bombar, J.~D. Duarte, V.~Dunning, S.~Etinski, D.~van
  Gent, M.~Lequesne, W.~van~der Schoot, M.~Stevens, A.~Cryptologists, and
  Advisors, ``{The PQC Migration Handbook: Guidelines for Migrating to
  Post-Quantum Cryptography -- Revised and Extended Second Edition},'' TNO,
  Applied Cryptography and Quantum Algorithms and CWI, Cryptology Group and
  AIVD, Netherlands National Communications Security Agency, Tech. Rep.,
  December 2024. [Online]. Available:
  \url{https://publications.tno.nl/publication/34643386/fXcPVHsX/TNO-2024-pqc-en.pdf}
\BIBentrySTDinterwordspacing

\bibitem{ott2023research}
\BIBentryALTinterwordspacing
D.~Ott, K.~Paterson, and D.~Moreau, ``{Where Is the Research on Cryptographic
  Transition and Agility?}'' \emph{Communications of the ACM}, vol.~66, no.~4,
  pp. 29--32, 2023. [Online]. Available: \url{https://doi.org/10.1145/3567825}
\BIBentrySTDinterwordspacing

\bibitem{naether24sok}
\BIBentryALTinterwordspacing
C.~Näther, D.~Herzinger, J.-P. Steghöfer, S.-L. Gazdag, E.~Hirsch, and
  D.~Loebenberger, ``{SoK: Towards a Common Understanding of Cryptographic
  Agility},'' 2024. [Online]. Available: \url{https://arxiv.org/abs/2411.08781}
\BIBentrySTDinterwordspacing

\bibitem{naether2024migrating}
\BIBentryALTinterwordspacing
C.~N{\"{a}}ther, D.~Herzinger, S.-L. Gazdag, J.-P. Stegh{\"{o}}fer, S.~Daum,
  and D.~Loebenberger, ``{Migrating Software Systems Towards Post-Quantum
  Cryptography -- A Systematic Literature Review},'' \emph{IEEE Access}, pp.
  132\,107--132\,127, 2024. [Online]. Available:
  \url{https://doi.org/10.1109/ACCESS.2024.3450306}
\BIBentrySTDinterwordspacing

\bibitem{wiesmaier2021pqc}
A.~Wiesmaier, N.~Alnahawi, T.~Grasmeyer, J.~Geißler, A.~Zeier, P.~Bauspieß,
  and A.~Heinemann, ``{On PQC Migration and Crypto-Agility},'' 2021.

\bibitem{ott2019identifying}
\BIBentryALTinterwordspacing
D.~Ott, C.~Peikert, and other~workshop participants, ``{Identifying Research
  Challenges in Post Quantum Cryptography Migration and Cryptographic
  Agility},'' 2019. [Online]. Available: \url{https://arxiv.org/abs/1909.07353}
\BIBentrySTDinterwordspacing

\bibitem{netwie24}
\BIBentryALTinterwordspacing
N.~von Nethen, A.~Wiesmaier, N.~Alnahawi, and J.~Henrich, ``{PMMP-PQC Migration
  Management Process},'' in \emph{Proceedings of the 2024 European
  Interdisciplinary Cybersecurity Conference}, ser. EICC '24.\hskip 1em plus
  0.5em minus 0.4em\relax New York, NY, USA: Association for Computing
  Machinery, 2024, pp. 144--154. [Online]. Available:
  \url{https://doi.org/10.1145/3655693.3655719}
\BIBentrySTDinterwordspacing

\bibitem{schhen24}
\BIBentryALTinterwordspacing
N.~Schmitt, J.~Henrich, D.~Heinz, N.~Alnahawi, and A.~Wiesmaier, ``{On Criteria
  and Tooling for Cryptographic Inventories},'' in \emph{{GI Sicherheit
  2024}}.\hskip 1em plus 0.5em minus 0.4em\relax Gesellschaft f{\"{u}}r
  Informatik e.V., 2024, pp. 49--63. [Online]. Available:
  \url{https://doi.org/10.18420/sicherheit2024\_003}
\BIBentrySTDinterwordspacing

\bibitem{etsi24cyberqsc}
{ETSI Technical Committee}, ``{ETSI TR QSC 0024 v0.0.6: CYBER -- A Repeatable
  Framework for Quantum-Safe Migrations},'' August 2024.

\bibitem{odonoghue2025softwarematerialssoftwaresupply}
\BIBentryALTinterwordspacing
E.~O'Donoghue, Y.~Hastings, E.~Ortiz, and A.~R.~M. Muneza, ``{Software Bill of
  Materials in Software Supply Chain Security A Systematic Literature
  Review},'' 2025. [Online]. Available: \url{https://arxiv.org/abs/2506.03507}
\BIBentrySTDinterwordspacing

\bibitem{leirimaa2024supporting}
\BIBentryALTinterwordspacing
K.~Leirimaa, ``{Supporting PQC migration and cryptographic agility with
  automated CBOM generation},'' University of Jyväskylä, Master thesis, 2024.
  [Online]. Available: \url{https://urn.fi/URN:NBN:fi:jyu-202411016889}
\BIBentrySTDinterwordspacing

\bibitem{hess24cboms}
B.~Hess and N.~K{\"o}rtge, ``{Standardization of Cryptography Bill of Materials
  in OWASP CycloneDX},'' in \emph{ETSI/IQC Quantum Safe Cryptography
  Conference}, 2024.

\bibitem{ma2021caraf}
\BIBentryALTinterwordspacing
C.~Ma, L.~Colon, J.~Dera, B.~Rashidi, and V.~Garg, ``{{CARAF: Crypto Agility
  Risk Assessment Framework}},'' \emph{Journal of Cybersecurity}, vol.~7,
  no.~1, p. tyab013, 05 2021. [Online]. Available:
  \url{https://doi.org/10.1093/cybsec/tyab013}
\BIBentrySTDinterwordspacing

\bibitem{hassim24framework}
K.~Hasan, L.~Simpson, R.~Baee, C.~Islam, Z.~Rahman, W.~Armstrong,
  P.~Gauravaram, and M.~McKague, ``{{A Framework for Migrating to Post-Quantum
  Cryptography: Security Dependency Analysis and Case Studies}},'' \emph{IEEE
  Access}, vol.~12, pp. 23\,427--23\,450, 2024.

\bibitem{pit97b}
\BIBentryALTinterwordspacing
J.~Pitman, ``{Some Probabilistic Aspects of Set Partitions},'' \emph{The
  American Mathematical Monthly}, vol. 104, no.~3, pp. 201--209, 1997.
  [Online]. Available: \url{https://doi.org/10.1080/00029890.1997.11990624}
\BIBentrySTDinterwordspacing

\bibitem{odlric85}
\BIBentryALTinterwordspacing
A.~Odlyzko and L.~Richmond, ``{On the number of distinct block sizes in
  partitions of a set},'' \emph{Journal of Combinatorial Theory, Series A},
  vol.~38, no.~2, pp. 170--181, 1985. [Online]. Available:
  \url{https://doi.org/10.1016/0097-3165(85)90066-4}
\BIBentrySTDinterwordspacing

\bibitem{sac97}
\BIBentryALTinterwordspacing
V.~Sachkov, \emph{{Probabilistic Methods in Combinatorial Analysis}}, ser.
  Encyclopedia of Mathematics and its Applications.\hskip 1em plus 0.5em minus
  0.4em\relax Cambridge University Press, 1997. [Online]. Available:
  \url{https://doi.org/10.1017/CBO9780511666193}
\BIBentrySTDinterwordspacing

\bibitem{pit97}
\BIBentryALTinterwordspacing
B.~Pittel, ``{Random Set Partitions: Asymptotics of Subset Counts},''
  \emph{Journal of Combinatorial Theory, Series A}, vol.~79, no.~2, pp.
  326--359, 1997. [Online]. Available:
  \url{https://doi.org/10.1006/jcta.1997.2791}
\BIBentrySTDinterwordspacing

\bibitem{erdren1960evolution}
\BIBentryALTinterwordspacing
P.~Erd{\H{o}}s and A.~R{\'e}nyi, ``{On the evolution of random graphs},''
  \emph{Publ. math. inst. hung. acad. sci}, vol.~5, no.~1, pp. 17--60, 1960.
  [Online]. Available: \url{https://doi.org/10.1515/9781400841356.38}
\BIBentrySTDinterwordspacing

\bibitem{frie23-random-graphs}
\BIBentryALTinterwordspacing
A.~Frieze and M.~Karo{\'n}ski, \emph{{Introduction to Random Graphs}}.\hskip
  1em plus 0.5em minus 0.4em\relax First edition: Cambridge University Press,
  July 2024, latest version available online
  \url{https://www.math.cmu.edu/~af1p/BOOK.pdf}. [Online]. Available:
  \url{https://doi.org/10.1017/CBO9781316339831}
\BIBentrySTDinterwordspacing

\bibitem{erdren59random}
\BIBentryALTinterwordspacing
P.~Erd{\H{o}}s and A.~R{\'e}nyi, ``{On random graphs I},'' \emph{Publ. math.
  debrecen}, vol.~6, no.~12, pp. 290--297, 1959. [Online]. Available:
  \url{https://doi.org/10.5486/PMD.1959.6.3-4.12}
\BIBentrySTDinterwordspacing

\bibitem{nist-sp-800-57-pt1-r5}
\BIBentryALTinterwordspacing
E.~Barker, ``{Recommendation for Key Management: Part 1 -- General (Revision
  5)},'' National Institute of Standards and Technology (NIST), Tech. Rep. NIST
  SP 800-57 Part 1 Rev. 5, May 2020,
  \url{https://doi.org/10.6028/NIST.SP.800-57pt1r5}. [Online]. Available:
  \url{https://csrc.nist.gov/publications/detail/sp/800-57-part-1/rev-5/final}
\BIBentrySTDinterwordspacing

\bibitem{Melton2007classcycles}
H.~Melton and E.~Tempero, ``{An Empirical Study of Cycles among Classes in
  Java},'' \emph{Empirical Software Engineering}, vol.~12, no.~4, pp. 389--415,
  2007.

\bibitem{Genfer2021cyclicmicroservices}
P.~Genfer and U.~Zdun, ``{Identifying Domain-Based Cyclic Dependencies in
  Microservice APIs Using Source Code Detectors},'' in \emph{Software
  Architecture: 15th European Conference, ECSA 2021, Virtual Event, Sweden,
  September 13–17, 2021, Proceedings}, ser. Lecture Notes in Computer
  Science, vol. 12891.\hskip 1em plus 0.5em minus 0.4em\relax Springer, 2021,
  pp. 207--222.

\bibitem{Evans2004dddACL}
E.~Evans, \emph{{Domain-Driven Design: Tackling Complexity in the Heart of
  Software}}.\hskip 1em plus 0.5em minus 0.4em\relax Boston, MA: Addison-Wesley
  Professional, 2004, introduces the Anti-Corruption Layer
  pattern—implemented as a façade/adapter—to shield a bounded context and
  eliminate circular dependencies.

\bibitem{alcaraz2015critical}
C.~Alcaraz and S.~Zeadally, ``{Critical infrastructure protection: Requirements
  and challenges for the 21st century},'' \emph{International journal of
  critical infrastructure protection}, vol.~8, pp. 53--66, 2015.

\bibitem{Kostyuk2019cyberattacks}
N.~Kostyuk and Y.~M. Zhukov, ``{Invisible Digital Front: Can Cyber Attacks
  Shape Battlefield Events?}'' \emph{Journal of Conflict Resolution}, vol.~63,
  no.~2, pp. 317--347, 2019.

\bibitem{tatipatri2024comprehensive}
N.~Tatipatri and S.~Arun, ``{A comprehensive review on cyber-attacks in power
  systems: Impact analysis, detection, and cyber security},'' \emph{IEEE
  Access}, vol.~12, pp. 18\,147--18\,167, 2024.

\bibitem{ieee802.16e-2005}
``{IEEE Standard for Local and Metropolitan Area Networks---Part 16: Air
  Interface for Fixed and Mobile Broadband Wireless Access Systems},'' IEEE,
  International Standard IEEE Std 802.16e-2005, Feb. 2006, published 28 Feb
  2006.

\bibitem{iec61968}
\BIBentryALTinterwordspacing
{IEC, TC 57 Power systems management and associated information exchange},
  ``{IEC~61968: Application integration at electric utilities\,--\,System
  interfaces for distribution management},'' {International Electrotechnical
  Commission (IEC)}, International Standard, 2020. [Online]. Available:
  \url{https://webstore.iec.ch/en/publication/32542}
\BIBentrySTDinterwordspacing

\bibitem{iec62351}
\BIBentryALTinterwordspacing
------, ``{IEC~62351: Power systems management and associated information
  exchange\,--\,Data and communications security},'' {International
  Electrotechnical Commission (IEC)}, International Standard, 2025. [Online].
  Available: \url{https://webstore.iec.ch/en/publication/6912}
\BIBentrySTDinterwordspacing

\end{thebibliography}

\begin{IEEEbiography}[{\includegraphics[width=1in,height=1.25in,clip,keepaspectratio]{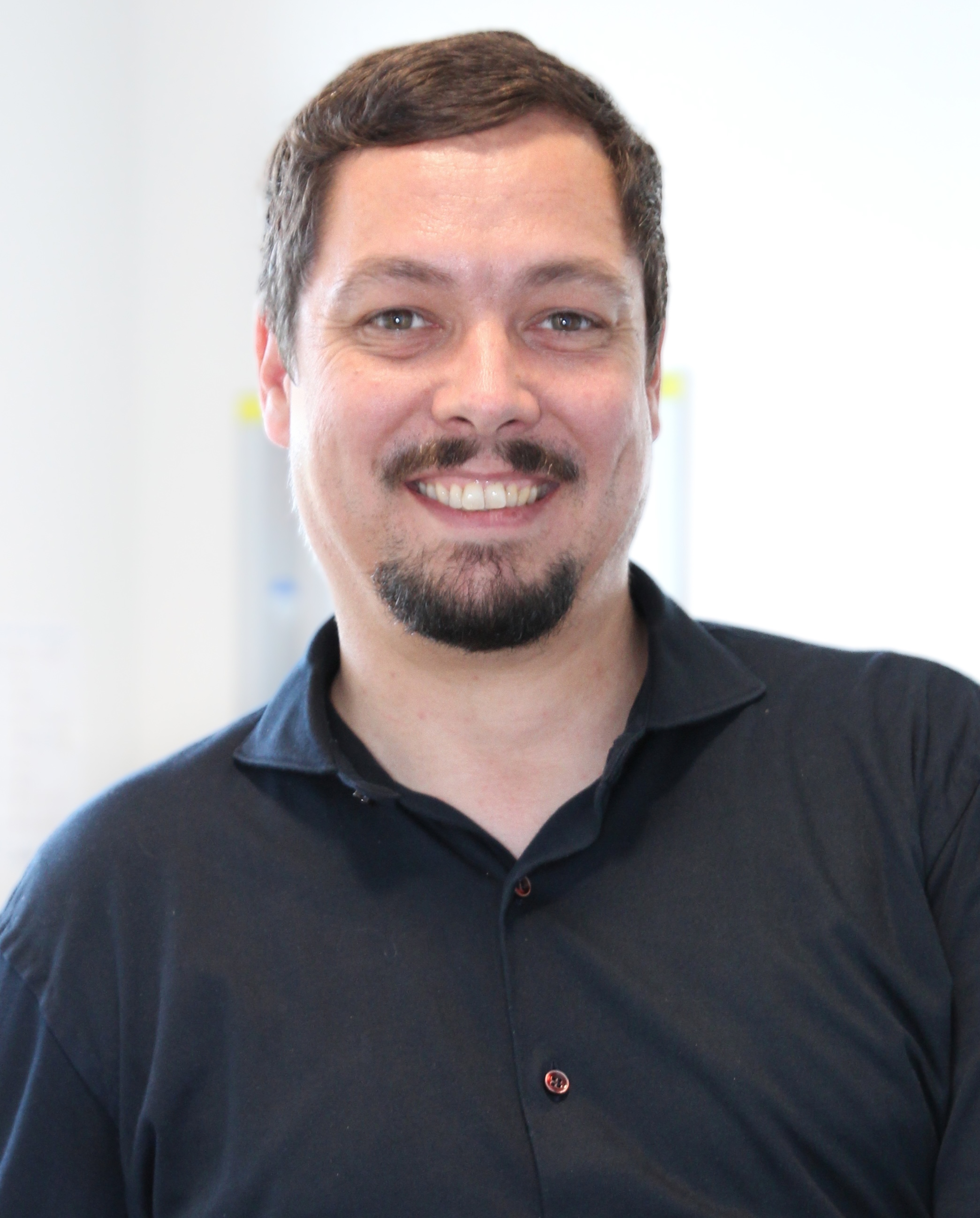}}]
{Daniel Loebenberger} received his doctorate in cryptography from the University of Bonn in 2012, where he conducted research and taught until the end of 2015. From 2016 to 2019, he worked as an IT security expert with a focus on cryptography at genua GmbH, a subsidiary of the Bundesdruckerei GmbH (German Federal Printing Office) specializing in high-security network solutions, on various topics of professional high-security components. Since January 2019, Daniel Loebenberger has been appointed Fraunhofer-endowed Professor of Cybersecurity at Fraunhofer AISEC in collaboration with the University of Applied Sciences Amberg-Weiden and also heads the department ``Secure Infrastructure'' at Fraunhofer AISEC's Weiden site. In the lab there, topics of applied post-quantum cryptography and quantum-safe infrastructures are addressed in research and teaching.
\end{IEEEbiography}

\begin{IEEEbiography}[{\includegraphics[width=1in,height=1.25in,clip,keepaspectratio]{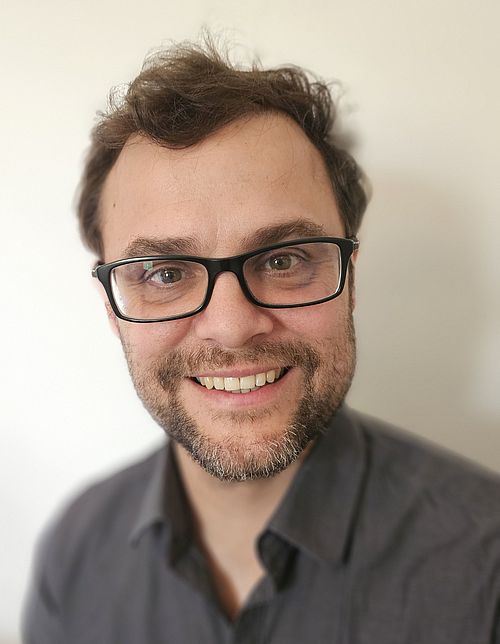}}]
{Eduard Hirsch} holds a Master of Science degree in Computer Science from
the Technical University of Vienna. Over the years he accumulated extensive experience in industry and contributing to various projects in the software and technology sector. In addition to his professional work, he has been teaching and doing research at the Salzburg University of Applied Sciences, where he was contributing to national (Austria) and international research projects (e.g. H2020, Interreg). His research interests span software engineering, industrial automation, big data engineering as well as security. He is currently a researcher at the
Technical University Amberg-Weiden, where he focuses on advancing both practical and theoretical aspects of PQC migration but also teaches in the area of security.
\end{IEEEbiography}

\begin{IEEEbiography}[{\includegraphics[width=1in,height=1.25in,clip,keepaspectratio]{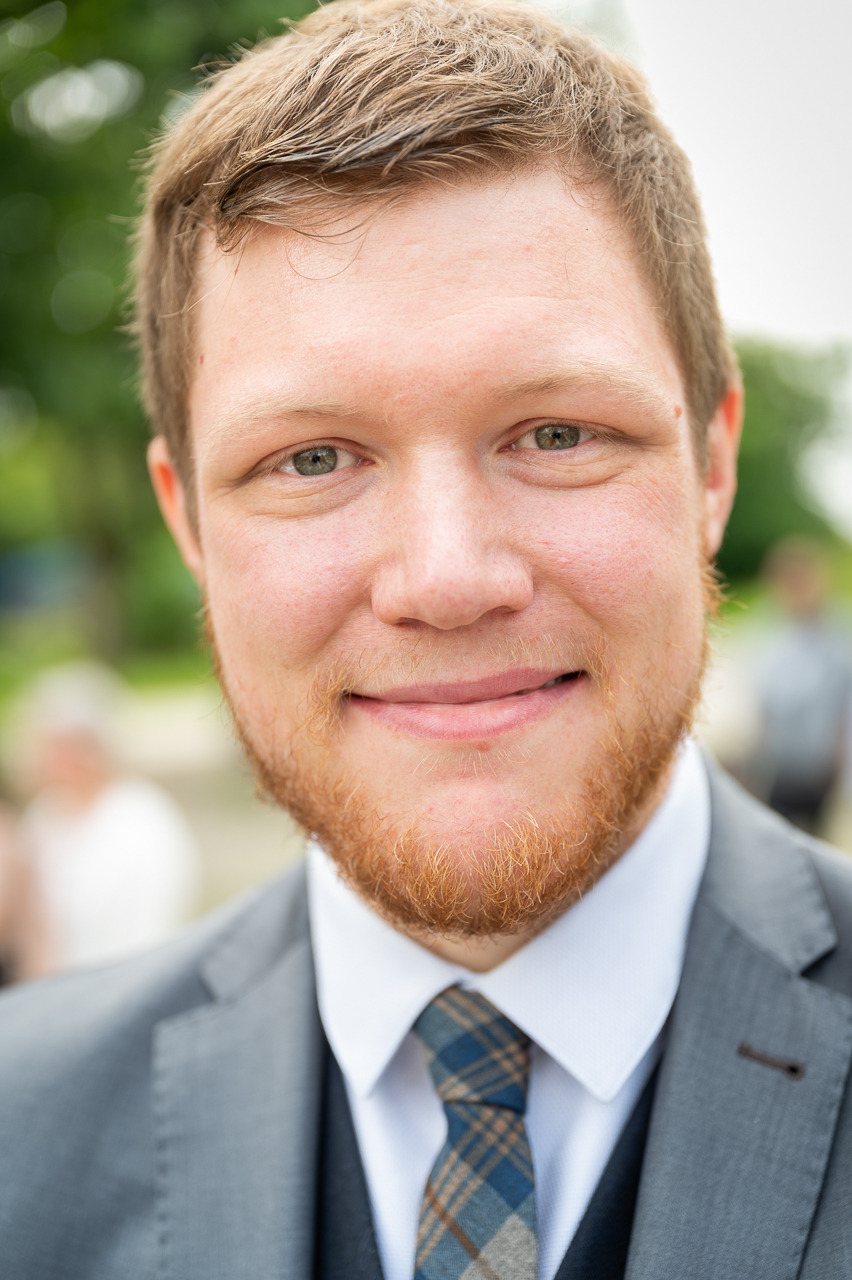}}]
{Stefan-Lukas Gazdag} holds a Master of Science in Computer Science. He focuses on post-quantum cryptography (PQC) and future-proof security mechanisms since 2013, overcoming obstacles in practical use. Being part of the German IT security company genua’s research group he currently works on the third publicly funded PQC research project and other collaborations, enabling the transition to quantum-safe networks. He was e.g. involved in the publication of the first explicit PQC RFC, RFC 8391 describing the signature scheme XMSS, recommended e.g. by US NIST, NSA and German BSI.
\end{IEEEbiography}
  
\begin{IEEEbiography}[{\includegraphics[width=1in,height=1.25in,clip,keepaspectratio]{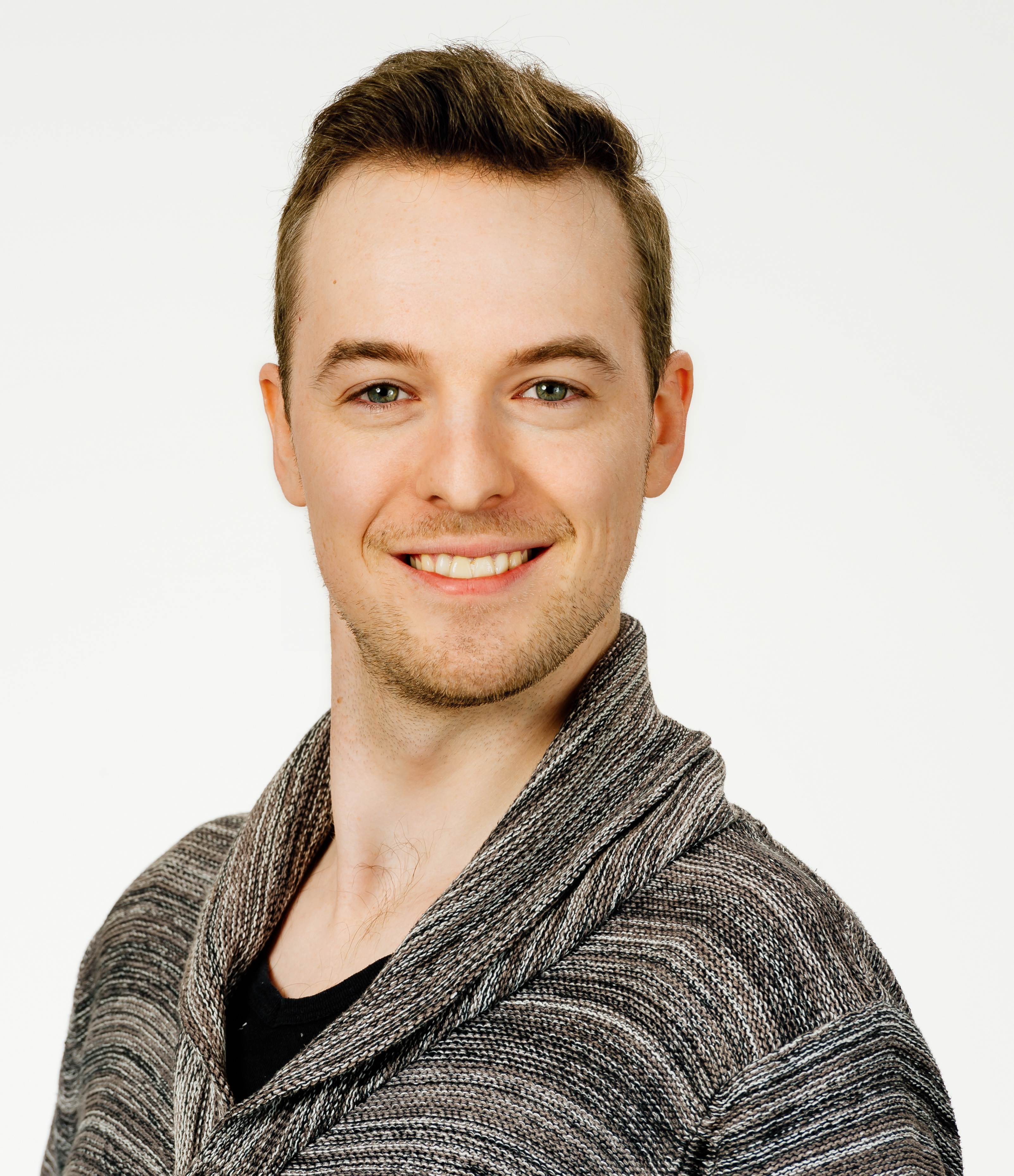}}]
{Daniel Herzinger} holds a Master of Science in Computer Science.
He started working on post-quantum cryptography (PQC) and their adoption into real-world scenarios during his Master's thesis in 2020, together with the German IT security company genua, where he started his IT security journey in 2015. After finishing his studies, he first dived into supporting customers with securing their IT and OT infrastructures at genua's consulting department for two years. He then indulged in his passion for PQC by joining the research group in 2023, focusing on the transition of complex systems and infrastructures to quantum resistance.
\end{IEEEbiography}

\begin{IEEEbiography}[{\includegraphics[width=1in,height=1.25in,clip,keepaspectratio]{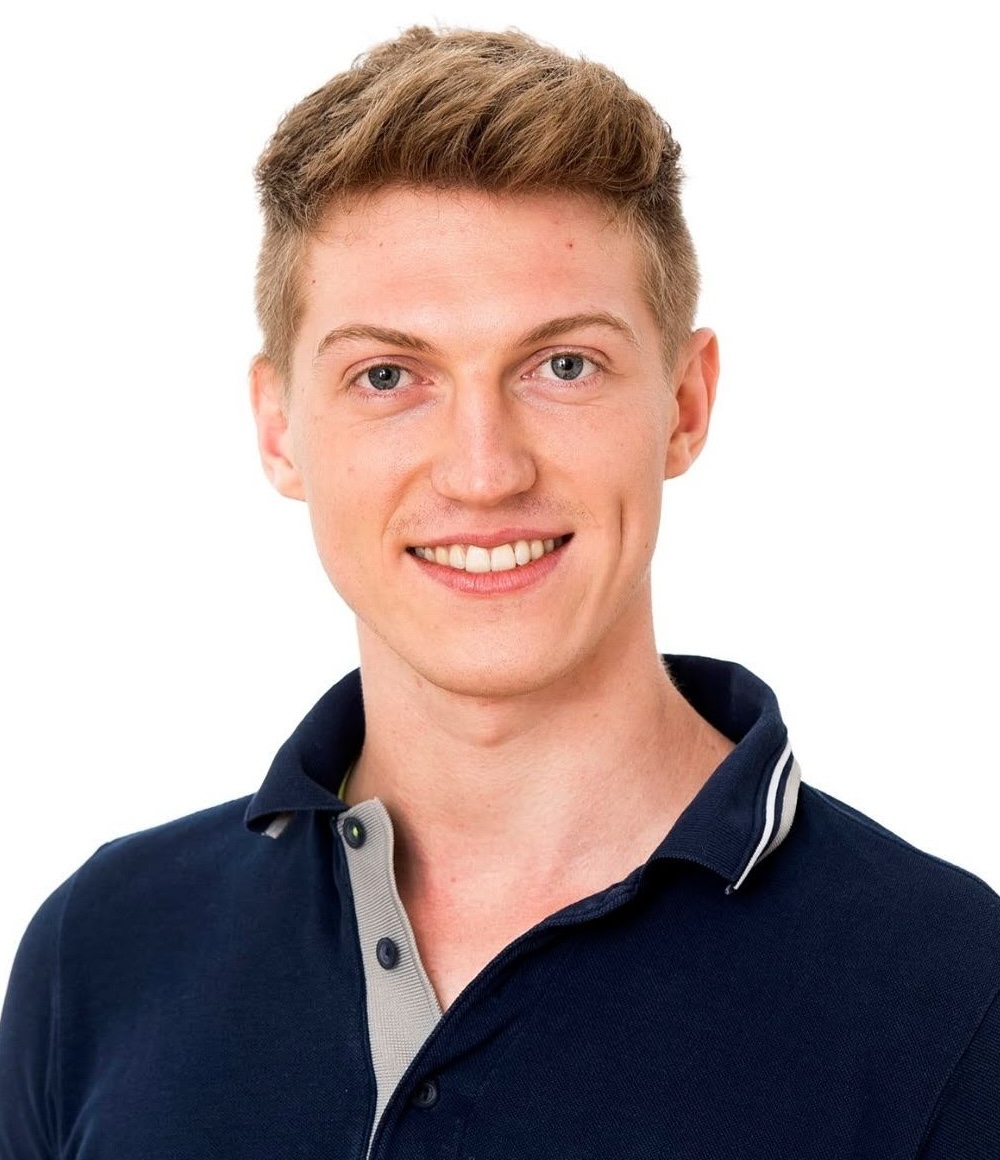}}]
  {Christian Näther} received his Master's degree in Computer Science from the University of Augsburg, Germany.
  As a former software engineer he has gained several years of experience in the fields of software development, distributed systems as well as information security. 
  He is currently a security researcher at XITASO and part of the research project AMiQuaSy, which focuses on the migration of systems towards post-quantum cryptography.
  Christian Näther is an active board member of XITASO’s information security community, where he promotes both internal and customer-specific security. 
\end{IEEEbiography}

\begin{IEEEbiography}[{\includegraphics[width=1in,height=1.25in,clip,keepaspectratio]{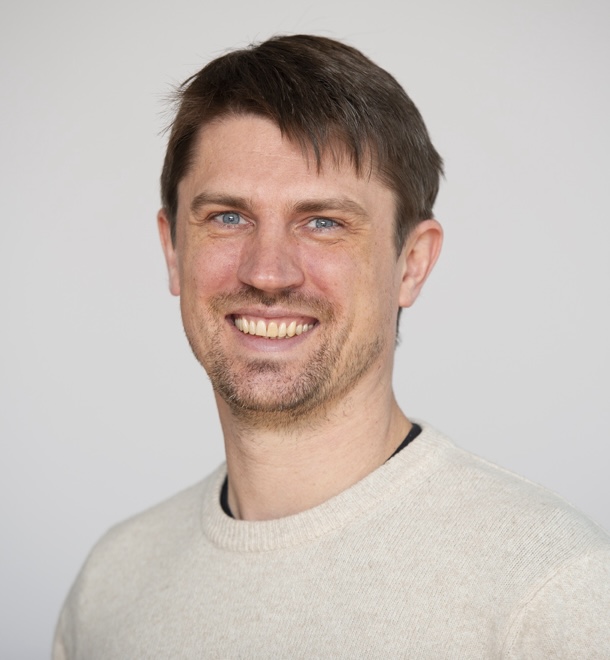}}]
{Jan-Philipp Steghöfer} (Member, IEEE) received the Ph.D. degree in Computer Science from the University of Augsburg, Germany.
As a former professor at the University of Gothenburg he has an extensive and distinguished academic background with a strong passion for research and innovation.
His research interests include security assurance, security of AI systems, agile development in regulated domains, and model-driven engineering.
He is currently a Senior Researcher at XITASO as well as an active member of the software engineering community where he is part of the organization committee for events like the IEEE International Requirements Engineering Conference (RE). 
He also reviews for International Conference on Software Engineering (ICSE), IEEE Transactions on Software Engineering (TSE), Journal of Systems and Software (JSS), Empirical Software Engineering and Measurement (ESEM), and many other venues.
\end{IEEEbiography}


\end{document}